\documentclass[11pt]{article}
\usepackage[utf8]{inputenc}
\usepackage[T1]{fontenc}
\usepackage{graphicx}
\usepackage{grffile}
\usepackage{longtable}
\usepackage{wrapfig}
\usepackage{rotating}
\usepackage[normalem]{ulem}
\usepackage{amsmath}
\usepackage{textcomp}
\usepackage{amssymb}
\usepackage{capt-of}
\usepackage{hyperref}
\usepackage{listingsutf8}
\usepackage{authblk}
\usepackage[numbers]{natbib}
\usepackage{subcaption}
\usepackage{braket}
\usepackage{alltt}
\usepackage{mathabx}
\usepackage{amsthm}
\usepackage[utf8]{inputenc}
\usepackage{verbatim}
\usepackage{stmaryrd}
\usepackage{mathtools}
\usepackage{xparse}
\usepackage{dsfont}
\usepackage{cleveref}
\crefname{lemma}{Lemma}{Lemmas}
\crefname{enumi}{}{}
\newtheorem{theorem}{Theorem}[section]
\newtheorem{lemma}{Lemma}[section]
\DeclarePairedDelimiterX\st[2]{\{}{\}}{\,#1 \;\delimsize\vert\; #2\,}
\DeclarePairedDelimiter\semant{\llbracket}{\rrbracket}
\author[ ]{Ryan Bernstein}
\affil[ ]{Columbia University}
\affil[ ]{\textit{ryan.bernstein@columbia.edu}}
\date{\today}
\title{Multi-Model Probabilistic Programming}
\hypersetup{
 pdfauthor={},
 pdftitle={Multi-Model Probabilistic Programming},
 pdfkeywords={},
 pdfsubject={},
 pdfcreator={Emacs 28.1 (Org mode 9.3.7)}, 
 pdflang={English}}
\begin{document}

\maketitle

\begin{abstract}
Probabilistic programming makes it easy to represent a probabilistic model as a program.
Building an individual model, however, is only one step of probabilistic modeling.
The broader challenge of probabilistic modeling is in understanding and navigating spaces of alternative models.
There is currently no good way to represent these spaces of alternative models, despite their central role.

We present an extension of probabilistic programming that lets each program represent a network of interrelated probabilistic models.
We give a formal semantics for these \emph{multi-model probabilistic programs}, a collection of efficient algorithms for network-of-model operations, and an example implementation built on top of the popular probabilistic programming language \emph{Stan}.

This network-of-models representation opens many doors, including search and automation in model-space, tracking and communication of model development, and explicit modeler degrees of freedom to mitigate issues like \(p\text{-hacking}\).
We demonstrate automatic model search and model development tracking using our Stan implementation, and we propose many more possible applications.
\end{abstract}
\pagebreak
\tableofcontents
\pagebreak

\cleardoublepage
\section{Introduction}
\label{sec:org92a967f}
\label{intro}
A probabilistic program is a program that represents a probabilistic model.
Probabilistic programming suggests an analogy between software engineering and probabilistic modeling: developing a model is like developing a program.
However, there is a key difference: while there are only ever a handful of relevant versions of a program in development, there are often a multitude of alternative probabilistic models that remain relevant throughout development, communication and validation.
Probabilistic programming systems currently ignore this multiplicity.

We present an extension of probabilistic programming that lets users encode \emph{networks of models}, graphs of models connected by similarity.
Just as probabilistic programs let users represent and query probabilistic models, these \emph{multi-model probabilistic programs} let users represent and query networks of probabilistic models.
\subsection{Uses of networks of models}
\label{sec:orgec79dc2}
   \label{uses}
We focus on four categories of use cases for networks of probabilistic models:

\begin{enumerate}
\item \textbf{Automation in model-space}
\label{automation}

Networks of models allow us to deploy automation in model-space.
For example, algorithms can search for models or neighborhoods in the network that optimize a real-valued model scoring function.
We demonstrate a greedy graph search in \cref{birthday} and more advanced search methods are discussed in \cref{future-applications}.
Other examples of automation on networks include:
\begin{itemize}
\item \emph{Stacking} methods seek an optimal weighted subset of the models to combine into a high-performing ensemble or ``stacked'' model \cite{stacking}.
\item \emph{Projection prediction} methods search for the ``simplest'' model that reproduces the predictions of a gold standard model or ensemble \cite{projpred}.
\end{itemize}
\item \textbf{Understanding the problem and its solution space.}
\label{space}

Statisticians often use multiple related models to gain insight into their problem and solution space \cite{bwf}.
Examples include:
\begin{itemize}
\item Plotting alternative models by evaluation metrics to understand trade-offs such as model complexity vs. accuracy \cite{loo}.
\item Comparing diagnostic samples, such as posterior predictive samples, to understand the impact of model decisions \cite{postpred}.
\item Applying ``multiverse'' methods, which seek to quantify modeling uncertainty by sampling from a whole set of plausible models \cite{multiverse}.
\end{itemize}
\item \textbf{Tracking and communicating the branching path of development.}
\label{tracking}

Probabilistic model development is an inherently iterative and branching process: models are improved by a cycle of criticism and adjustment \cite{bwf,box}.
Development follows an often-backtracking path through model space; a path that model developers often have difficulty managing \cite{guo,patel,variolite}.

When the relevant region of model space is explicitly represented as a network of models, the path of development can also be explicitly documented.
This documentation is useful for managing development \cite{variolite}, for third parties interested in learning from or extending the work, and, as discussed in case \cref{dof}, as context for third-party auditors assessing the quality of the final model.

Standard version control tools like Git can serve this use case to some extent, but as the case study in \cref{golf} shows, the network of models abstraction earns more than version control: it keeps all relevant models available and tracks their semantic relationships.
Two surveys of data scientists, one by \citeauthor{guo} and one by \citeauthor{variolite}, both report that version control is not typically used for managing exploratory data analysis \cite{guo,variolite}.
\item \textbf{Making modeler degrees of freedom explicit.}
\label{dof}

\(P\text{-hacking}\) \cite{p-hacking} and the Garden of Forking Paths \cite{garden-of-forking-paths} are issues in statistics with a common root cause called \emph{modeler degrees of freedom} \cite{degrees-of-freedom}: when a modeling task, such as the analysis of scientific data, includes modeling decisions with multiple justifiable solutions, the modeler can, intentionally or not, tune their solutions to cherry-pick a desirable model, such as a model with a ``significant'' \(p\text{-value}\).

There are two ways in which explicit networks of models can alleviate the issues of modeler degrees of freedom.
The first way is to aid in reporting the set of analyses that were done.
\Citeauthor{playing}, discussing questionable research practices, argue: ``the potentially (highly) questionable part of your actions as a researcher is not that you engage in all kinds of exploratory analyses. Instead, the questionable part is not reporting truthfully and explicitly the exploratory nature of these analyses'' \cite{playing}.
When researchers explicitly report their path of exploration through the network of models, their degrees of freedom become transparent.

The second way is to automate sensitivity analysis.
Sensitivity analyses aim to ``assess whether altering any of the assumptions made leads to different final interpretations or conclusions'' \cite{leukaemia}.
They are broadly recommended for scientific reports such as clinical trials \cite{sensitivity-clinical-trials}.
In a survey of solutions to researcher degrees of freedom, \citeauthor{four-solutions} finds that sensitivity analyses ``provide an effective solution to the \(p\text{-hacking}\) problem'', and also the Forking Paths problem in some cases\footnote{Rubin distinguishes between result-biased and result-neutral Forking Paths, and finds that sensitivity analysis is not necessarily sufficient for result-neutral cases.} \cite{four-solutions}.
Explicit networks of models let both researchers and third-party auditors automate sensitivity analysis by a simple procedure: for each conclusion drawn from a final model, check the extent to which that conclusion is also drawn by the model's neighbors.
When each neighbor differs by one modeling decision, the analysis tells us the modeling decisions on which the conclusions depend.
\end{enumerate}

Some of these use cases, like automated model search and visualizations of model space, are achieved in practice only by those with sufficient time and expertise to implement ad-hoc methods with hand-enumerated sets of models.
Other use cases, like explicit modeler degrees of freedom and automated sensitivity analysis, are rarely ever achieved in practice despite being valuable in theory.
We argue that all of these use cases could become convenient and routine if we had a standard representation of networks of models.

We note that each of the four use cases above has a natural definition and utility for drawing edges between models:
\begin{itemize}
\item For use case \cref{automation} (automation), edges should be between the most similar models: the network is then more analogous to a continuous and differentiable space, and methods like greedy graph search more closely approximate gradient descent.
\item For use case \cref{space} (solution space mapping), again edges should be between the most semantically similar models; edges then provide a more consistent sense of distance and orientation.
\end{itemize}

\begin{itemize}
\item For use case \cref{tracking} (tracking and communicating development), edges should bridge sequential versions of a model. 
Authors and auditors then can cleanly trace model development, with each of its decision points and model transformations, as a tree within the network.
\item For use case \cref{dof} (explicit degrees of freedom), edges should be between model pairs that are that are one ``decision'' apart; then each edge is like a step in the Garden of Forking Paths.
\end{itemize}

Ideally, a standard construction of model networks should be compatible with all of these use cases.

\subsection{Representing networks of models}
\label{sec:org0d0d45a}
To support all of the above use cases, representations of networks of models should support the following operations efficiently:
\begin{enumerate}
\item \label{usecase:network} \textbf{Explicitly generate the network of models (if practical).}
\item \label{usecase:neighbors} \textbf{Given a node, generate its set of neighbors.} This operation becomes necessary when the network is too large to practically generate in its entirety; for instance, when we are searching through a large model space.
\item \label{usecase:concretize} \textbf{Given a node, generate its corresponding probabilistic program.}
\end{enumerate}

\noindent In addition, we would like our representation to be:
\begin{enumerate}
\item \textbf{Easy to read and understand.} Especially for cases \cref{tracking,dof}, clarity is the priority.
\item \textbf{Simple to write.} Probabilistic modelers may not be expert programmers.
\item \textbf{Scalable to many models.} If a representation is too redundant, large numbers of models become cumbersome.
\item \textbf{Standardized.}

A standard format enables a stable network-of-models API to easily build general model-space tools.
\end{enumerate}

One obvious choice for representing multiple models is a directory of relevant probabilistic programs.
According to surveys done by \citeauthor{guo} and \citeauthor{variolite}, this is the typical solution among data scientists \cite{guo,variolite}.

While appealing for their simplicity, file collections are a poor solution because they become uninterruptible, unmanageable, and memory intensive as the number of models grows large, and because they discard the semantic relationships between models.
\citeauthor{variolite} also report that data scientists have issues with file naming, keeping track of the relationships between files, and maintaining a mental map of their code \cite{variolite}.

Another possible choice of representation is a program in a general-purpose language that generates the set of probabilistic programs.
This approach is flexible and scalable, but its flexibility also makes it prohibitively difficult to write, understand, and standardize.

Our proposed representation is a middle ground: a meta-programming feature to augment existing probabilistic programming languages, so that meta-programs, which we call multi-model programs, represent networks of probabilistic models.
We argue that our meta-programming approach is nearly as flexible as general program generation and easier to write and understand than the directory-of-programs approach for nontrivial examples.
We call our meta-programming feature \emph{swappable modules}.

\subsection{Swappable modules in Stan}
\label{sec:org115cfc7}

To demonstrate swappable modules, we use Stan as our example host language, mainly because Stan is popular, performant, and has a clear, established semantics \cite{stan-manual,slicstan}.
Stan is also a highly structured and restrictive language, so by adding a swappable modules to Stan, we are demonstrating an usually difficult case that can easily be transferred to other languages.

We refer to the Stan language augmented with swappable modules as \emph{modular Stan}.
We have built a prototype compiler \cite{mstan-github} and interactive visualization website \cite{website} for modular Stan.

\begin{figure}[htbp]
\centering
\includegraphics[width=400pt]{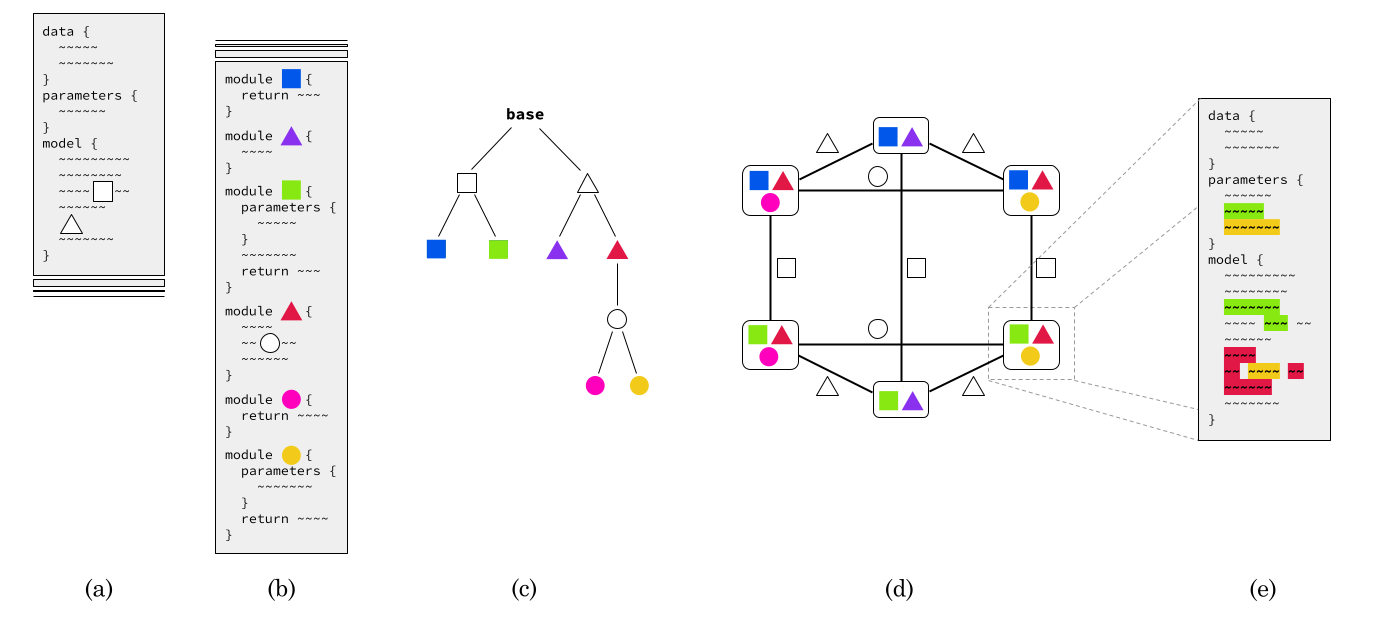}
\caption{\label{fig:color-shape}The correspondence between a modular Stan program and a network of models. Holes are represented as empty shapes. Modules are represented as filled shapes, with the same shape as the hole that they fill and a unique color. (a) and (b) represent a modular Stan program in its two parts, the base (a) and the module implementations (b). (c) represents the relationships between the base, holes and modules implementations. (d) represents the model graph, where each node is a model corresponding to a valid selection of modules, and each edge is labeled with the hole by which its endpoints differ. (e) shows the Stan program corresponding to a node; it is synthesized from the selected modules of the node.}
\end{figure}

\Cref{fig:color-shape} shows an abstracted modular Stan program and two data structures derived from it.

A modular Stan program is made of two parts.
The first part is the \emph{base}.
The base is like a Stan program, except that its code can contain \emph{holes}, which are syntactically similar to function calls but are more flexible.
The second part is a list of \emph{module implementations}.
Each module implementation describes a new way to ``fill'' a particular hole.
More than one module implementation be specified to ``fill'' the same hole.
Module implementations can themselves contain holes.

We can visualize the holes and implementations, with their implementation-fills-hole and code-contains-hole relationships, as a rooted directed acyclic graph of base, hole and module implementation nodes, where the root represents the program's base, each base and module implementation node points to the holes it contains, and each hole node points to the module implementations that fill it. We refer to this graph as the \emph{module graph} of a program.
The module graph is not necessarily a tree because multiple pieces of code can contain the same hole.

To build a valid Stan program out of a modular Stan program, we need to \emph{select} one module implementation to fill each hole in its base, and one implementation to fill each hole in those implementations, and so on, until there are no holes.
A minimal set of module implementations that leaves no empty holes is called a \emph{valid selection}.
Given a modular Stan program and a valid selection, we can produce a valid Stan program by filling each hole with the given module implementation.

A modular Stan program can therefore represent many Stan programs: one for each valid selection.
If we consider each of these Stan programs to be a node, we find a natural network structure: draw an edge between two Stan program nodes when their selections conflict by only one choice.
\Cref{fig:color-shape} shows such a graph with correspondence of each single node to a Stan program and thereby to a probabilistic model.
We argue that this graph construction is a powerful way to represent networks of models.

Holes and module implementations are like function calls and function definitions, but with important differences:
\begin{itemize}
\item Selecting a module implementation can add model parameters and have other changes that are not local to their call site.
Stan requires the set of parameters to be fixed at compile time; therefore modules are selected before Stan's compile time.
\item More than one module implementation can be specified to fill a hole, like function overloading but with identical type signatures.
\item Because module implementation code is essentially inlined statically, holes can appear in places where function calls cannot.
\end{itemize}

By defining a module system for Stan, we are also incidentally letting users write their normal Stan programs in a more modular way.
Modularity of this level is a feature that is both important for probabilistic modeling workflow \cite{bwf} and conspicuously absent from probabilistic programming languages like Stan \cite{bob-probprog}.

Probabilistic programming languages augmented with swappable modules meet all of our earlier criteria to represent networks of models:
\begin{itemize}
\item They are easy to read and write, because only one abstraction (the \texttt{module}) and a small amount of syntax is introduced on top of the host language, and the \texttt{module} abstraction already fits naturally into the probabilistic modeling workflow \cite{bwf}.
\item They are scalable, because the combinatory nature of implementation selection can define many models with little code. The language extensions in \cref{extensions} further increase that expressiveness.
\item They are standardized, because compilers enforce well-defined languages.
\item The required network-building, neighbor-finding and model-selection operations can be implemented efficiently.
\end{itemize}

In addition, the definition of edges as between models that differ by one choice module is compatible with all four uses cases' definitions of edges.

\bigskip
In this paper we introduce the syntax, semantics, and operation algorithms of a swappable module system for Stan.
For each algorithm we provide proofs of completeness, correctness and efficiency.

We also introduce language extensions and macros to conveniently express common patterns of model variation in terms of our module system.
We also introduce a macro system to concisely express common patterns of model variation and large, complex model families.

We demonstrate some of the benefits of multi-model probabilistic programs with two brief but real-world data science case studies.
\section{Related work}
\label{sec:orgeec4b64}
The 'network of models' is an established concept in statistics and machine learning, but explicit and general network of models abstractions are not supported within probabilistic programming languages or other statistical software. For instance, multi-model methods like ensemble methods and model search are common in machine learning, but are typically ad-hoc in that they do not start with a declarative network topology for their set of model. The closest method that we are aware of is perhaps grid search, which searches a pre-defined grid of model hyperparameters.

There are existing probabilistic programming systems which allow users to define their programs by combining probabilistic subcomponents. 
For example, Prophet \cite{prophet} allows users to combine subcomponents into time-series models. 
This effect could also be achieved in embedded probabilistic programming languages whose host languages have sufficiently expressive module systems. 
We are not aware of any other module systems specialized to encapsulate flexible components of structured probabilistic programs.

\Citeauthor{variolite} developed a tool for code editors, Variolite, to support exploratory data science by tracking alternative snippets of code in version control.
Variolite lets developers write, visualize and manage iterative and branching versions of data science pipelines in a similar way to our proposal.
Variolite addresses what we call use case \cref{tracking} (development tracking).
Variolite differs from our system primarily in that it is a tool for code editors rather than a metaprogramming feature, it does not produce a network of models, and because it is language agnostic, it only supports swapping out regions of code rather than more general semantic units.

There are existing systems that allow users to specify a program's components at compile-time.
Backpack \cite{backpack} is a build-system tool for the Haskell ecosystem that lets users swap out external software libraries that implement a common interface.
Much of work done by \citeauthor{backpack} to introduce Backpack as a mixin linker also applies to our swappable module system.
In practice, the C preprocessor is often used for this purpose; it can include or exclude sections of code depending on user flags or system environment properties.
These systems are not commonly applied to probabilistic programs or to study networks of programs.

\subsection{Comparison to ML-like module systems}
\label{sec:org9996ee5}
Our swappable module system bears some resemblance to ML-like module systems \cite{ml,ml-modules}. We find OCaml to be a helpful comparison point \cite{ocaml-modules}.
The \texttt{hole} and \texttt{module} approach can be understood in the language of ML-like modules: 

Each \texttt{hole} in a program, which is to say each unique hole identifier referenced in a program, can be thought of as declaring a module signature and a module-valued variable of that signature.
Each statement and expression referencing a \texttt{hole} is like a reference to a field of that \texttt{hole}'s corresponding variable.
A \texttt{hole}'s variable may take the value of any \texttt{module} that ``implements'' it.
A \texttt{hole}'s signature is inferred from its usage and implementations.
The value assigned to a \texttt{hole}'s variable may either be specified outside of the program or left non-deterministic.
Blocks of code that contain \texttt{hole}s can then be thought of as ML functors in that their \texttt{hole}s are like implicit module-valued arguments.

However, unlike ML-like modules:
\begin{enumerate}
\item Modules are not applied at any point in the program. Rather, the programs' semantic domain is the set of programs generated by any combination of module applications, as though module applications were non-deterministic; then the user can determine the modules to apply in an optional mode of compilation called \emph{concretization} that we discuss in \cref{concretize}.
\item When a \texttt{module} is assigned to a \texttt{hole}, or \emph{selected}, there can be global effects on the semantics of the program, because the \texttt{module} may add code to Stan's top-level blocks. For instance, a \texttt{module} may introduce a new model parameter, which changes the domain over which the program defines a joint distribution.
This effect is especially noteworthy in languages like Stan in which model parameters are fixed at compile time, because it implies that module application must happen before compile time.
\end{enumerate}

We are not aware of prior examples of inferred implicit module signatures, non-deterministic functor application semantics, or module application with non-local effects on the resulting program.

\section{Background: Stan}
\label{sec:org583f226}
  \label{stan-syntax}
Stan is our example host language, so we give a brief overview of Stan programs and of Stan's syntax in this section.

Like all probabilistic programs, Stan programs represent probabilistic models.
Stan programs are C-like, imperative, and are written as a sequence of top-level blocks.
Here is a simple example of a Stan program:

\begin{verbatim}
data {
    int N;
    vector[N] x;
}
parameters {
    real mu;
    real sigma;
}
model {
    mu ~ normal(0, 1);
    sigma ~ lognormal(0, 1);
    x ~ normal(mu, sigma);
}
\end{verbatim}

We see three of Stan's blocks: \texttt{data}, \texttt{parameters}, and \texttt{model}.
The \texttt{data} block declares the observed variables and the \texttt{parameters} block declares the unobserved variables. The \texttt{model} block define the log-density of the joint distribution of the observed and unobserved variables. Each \texttt{\textasciitilde{}} statement implicitly increments a variable \texttt{target} that represents the value of the overall log-density function; for instance, \texttt{mu \textasciitilde{} normal(0, 1);} could be rewritten \texttt{target += normal\_lpdf(mu, 0, 1);}.

The above program represents the probabilistic model:
 \begin{align*}
   &\mu \sim Normal(0, 1)\\
   &\sigma \sim LogNormal(0, 1)\\
   &x \sim Normal(\mu, \sigma)
\end{align*}
where \(x\) is an observed variable. When the program is compiled and executed given data for \(x\), it should produce samples from the posterior distributions \(P(mu | x)\) and \(P(sigma | x)\).

For our purposes, it is sufficiently precise to take the semantic domain of a Stan program \(p\), \(\semant{p}\), to be a joint distribution \(P(d, \theta)\) of the variables declared in the \texttt{data} block, \(d\), and \texttt{parameters} block, \(\theta\), of \(p\), from which can be inferred a posterior distribution \(P(\theta \mid d)\).

\subsection{Syntax}
\label{sec:org736e29a}
Stan programs are organized into blocks of statements that describe different aspects of a probabilistic model.
Each is an ordered subset of blocks:
\begin{verbatim}
STAN_PROG: FUNCTIONS?
           DATA?
           TRANSFORMED_DATA?
           PARAMETERS?
           TRANSFORMED_PARAMETERS?
           MODEL?
           GENERATED_QUANTITIES?
\end{verbatim}
Here \texttt{?} indicates that the block may or may not be present.

Below is the syntax of the \texttt{data} block:
\begin{verbatim}
DATA: data { STMT_DECL;* }
STMT_DECL: TYPE identifier
TYPE: int | real | vector | matrix | ...
\end{verbatim}
Here, \texttt{identifier} stands in for valid Stan variables names.
The \texttt{*} symbols are Kleene stars to indicate an element that can be repeated or absent.

The \texttt{parameters} block is similar\footnote{Strictly speaking, we should not allow discrete types like \texttt{int} to be declared as \texttt{parameters}.} :
\begin{verbatim}
PARAMETERS: parameters { STMT_DECL;* }
\end{verbatim}

The \texttt{model} block is more flexible:
\begin{verbatim}
MODEL: model { STMT_LPDF;* }
STMT_LPDF: STMT_BASIC
         | identifier ~ identifier( EXPR,* );
         | target+= EXPR;
STMT_BASIC: STMT_DECL | STMT_ASSIGNMENT | STMT_FOR
          | STMT_IFELSE | STMT_FUNCTION_APPLICATION | ...
\end{verbatim}
Here, \texttt{target} is a reserved variable in Stan that represents the accumulated log-value of the density function.

\texttt{STMT\_BASIC} and \texttt{EXPR} closely resemble C-like languages, so we omit their details here.

Stan allows user-defined function declared in the \texttt{functions} block:
\begin{verbatim}
FUNCTIONS: functions { FUNC_DECL* }
FUNC_DECL: RET_TYPE identifier ((TYPE identifier),*) { STMT_FUNC;* }
STMT_FUNC: STMT_BASIC | return EXPR
RET_TYPE: TYPE | void
\end{verbatim}

There are three more blocks: 
\begin{verbatim}
TRANSFORMED_DATA: transformed data { STMT_BASIC;* }
TRANSFORMED_PARAMETERS: transformed parameters { STMT_LPDF;* }
GENERATED_QUANTITIES: generated quantities { STMT_BASIC;* }
\end{verbatim}

The \texttt{transformed data} and \texttt{transformed parameters} blocks let users define transformed versions of the observed and hidden variables in a way that works efficiently with the inference process. The \texttt{generated quantities} block lets users define output quantities calculated from the samples of the parameters.

\subsection{Effects and Scope}
\label{sec:orgf546d73}
\label{stan-effects-scope}
Let \texttt{Block} be the set of Stan block types, \(\{\texttt{data}, \texttt{parameters}, \texttt{model}, \dots\}\).

Stan statements and expressions are sometimes allowed to be impure in particular ways depending which block contains that code.
We call those impurities \emph{effects.}
We say that when code uses the random number generator, it has the \texttt{RNG} effect, and when it increments the program's density function with \texttt{\textasciitilde{}} or \texttt{target+=} statements, it has the \texttt{LPDF} effect.
The set of effects \texttt{Eff} is then \(\{\texttt{RNG}, \texttt{LPDF}\}\).
Stan's specification implicitly defines a mapping \texttt{effects} from some \(block \in \texttt{Block}\) to the set of effects allowed within that block, \(\texttt{effects}(block)\subset \texttt{Eff}\).

In Stan programs, the declarations that a statement may reference depends on the statement's block and the declaration's block.
For instance, code in a \texttt{transformed data} block can reference top-level declarations in a \texttt{data} block but not in a \texttt{parameters} block.
To know whether it is valid to inserting new statements into a given block, therefore, we need to know to which blocks' declarations that statement is allowed to refer.
Stan's specification implicitly defines a mapping \texttt{scope} from some \(block \in \texttt{Block}\) to the set of blocks whose top-level declared variables statements that block may reference, \(\texttt{scope}(block) \subset \texttt{Block}\).

\subsection{Program validity}
\label{sec:orgd53f428}
\label{valid-stan}

We define \(valid_{\texttt{Stan}}\) to denote whether a program is \emph{valid}, or roughly whether we expect it to compile.
Let \(valid_\texttt{Stan}(SP)\) for a Stan program \(SP\) if and only if all of the following are true:
\begin{enumerate}
\item \textbf{Effects and scopes of code are available in their block}.
\item \textbf{\(SP\) typechecks as Stan code.}
\end{enumerate}

\section{Modular Stan syntax}
\label{sec:orgeca1ab8}
\label{mstan-syntax}

Below is an example modular Stan program:
\begin{verbatim}
data {
  int N;
  vector[N] x;
}
model {
  x ~ normal(Mean(), Stddev());
}

module "standard" Mean() {
  return 0;
}

module "standard" Stddev() {
  return 1;
}

module "normal" Mean() {
  parameters {
    real mu;
  }
  mu ~ normal(0, 1);
  return mu;
}

module "lognormal" Stddev() {
  parameters {
    real<lower=0> sigma;
  }
  sigma ~ lognormal(0, StddevInformative());
  return sigma;
}

module "yes" StddevInformative() {
  return 1;
}

module "no" StddevInformative() {
  return 100;
}
\end{verbatim}

In the above program, the base is made up of the \texttt{data} and \texttt{model} blocks, \texttt{Mean}, \texttt{Stddev} and \texttt{StddevInformative} are holes, and each block starting with the keyword \texttt{module} are the module implementations.

The above modular Stan program has similar structure to the example in \cref{fig:color-shape}.
It also includes the example program \cref{stan-syntax} as one of its nodes (where \texttt{Mean} is filled by \texttt{normal}, \texttt{Stddev} is filled by \texttt{lognormal}, and \texttt{StddevInformative} is filled by \texttt{yes}).

Modular Stan makes two additions to Stan's syntax: \emph{Holes} and \emph{module implementations.}

Holes are statements or expressions that are syntactically similar to function applications.
We define variants of the Stan syntax rules that are allowed to include holes.
For each Stan grammar rule \texttt{RULE} that can directly or indirectly contain a \texttt{STMT\_BASIC}, \texttt{STMT\_LPDF}, or \texttt{EXPR}, we define a new rule \texttt{RULE\_M} that replaces those rules with the following \texttt{STMT\_BASIC\_M}, \texttt{STMT\_LPDF\_M}, and \texttt{EXPR\_M} rules, respectively.
\begin{verbatim}
STMT_BASIC_M: hole_identifier(EXPR_M,*);
            | STMT_DECL | STMT_ASSIGNMENT_M | STMT_FOR_M
            | STMT_IFELSE_M | STMT_FUNCTION_APPLICATION_M | ...

STMT_LPDF_M: identifier ~ hole_identifier(EXPR_M,*); 
           | identifier ~ identifier( EXPR_M,* );
           | target+= EXPR_M;
           | STMT_BASIC_M
EXPR_M: hole_identifier( EXPR_M,* ) | ...
\end{verbatim}
We use \texttt{hole\_identifier} to stand in for valid hole names.

Module implementations are reminiscent of function definitions, and appear at the top level alongside blocks:

\begin{verbatim}
MODULE_IMPLEMENTATION_M:
  module "impl_identifier" hole_identifier((TYPE identifier,)*) {
    PARAMETERS?
    STMT_LPDF_M;*
    return EXPR_M;*
  }
\end{verbatim}
We use \texttt{impl\_identifier} to stand in for valid implementation names, while \texttt{module} is a new keyword.

A modular Stan program is then:

\begin{verbatim}
MODULAR_STAN_PROG: STAN_PROG_M
                   MODULE_IMPLEMENTATION_M*
\end{verbatim}

We make two small additions to the syntax and capabilities of modular Stan in \cref{extensions}, but this base syntax is sufficient to introduce its semantics and algorithms.
\section{Modular Stan semantics}
\label{sec:orgb3f8bd1}
\label{semantics}
\subsection{Basic operations on programs}
\label{sec:orgfd4863d}
We assume that there is some parsing procedure \texttt{Parse} such that, for all strings \(F\) of the language defined by the syntax in \cref{mstan-syntax}, \(P = \texttt{Parse}(F)\), where \(P\) is some reasonable representation of \(F\) that we refer to loosely as a ``modular Stan program''. 

While the actual representation of a program \(P\) is an implementation detail, we can think of \(P\) as effectively a pair \(P = (P_{base}, impls(P))\), where \(P_{base}\) represents the Stan-like \emph{base} of \(P\) and \(impls(P)\) is the set of all module implementations defined in \(P\).
\(P\) also implicitly includes the set of all holes referenced by \(P_{base}\) and \(impls(P)\).

We likewise are not concerned with the representations details of implementations or holes, we only need to define operations on them.

Below are the basic operations we use to interact with programs \(P\), implementations \(i\), holes \(h\), sets \(I\) of implementations, and sets \(H\) of holes:

\begin{description}
\item[{\(impls(h)\)}] is the set of implementations that implement a hole \(h\).
\item[{\(impls(P)\)}] is the set of all implementations defined in \(P\).
\item[{\(impls(H)\)}] \(= \bigcup_{h \in H} impls(h)\).
\item[{\(holes(i)\)}] is the set of holes referenced in the definition of an implementation \(i\).
\item[{\(holes(P_{base})\)}] is the set holes referenced in the base of \(P\).
\item[{\(holes(I)\)}] \(= \bigcup_{i \in I} holes(i)\).
\item[{\(holes(P)\)}] \(= holes(P_{base}) \cup holes(impls(P))\).
\item[{\(par(i)\)}] is the hole that the implementation \(i\) implements, also called the \emph{parent} of \(i\).
\item[{\(pars(I)\)}] \(= \bigcup_{i \in I} pars(i)\).
\item[{\(pars(P)\)}] \(= pars(impls(P))\).
\end{description}

The above operations are specific to the context of a program \(P\). Since the intended \(P\) is usually clear, we only give the operation a subscript when disambiguation is necessary.

\bigskip

We note that \(par\) and \(impls\) operations are like inverses, so:
$$i \in impls(h) \Leftrightarrow h \in pars(i)$$
$$h \in pars(I) \Leftrightarrow \exists i \in I\ s.t.\ i \in impls(h)$$
$$i \in impls(H) \Leftrightarrow \exists h \in H\ s.t.\ h = par(i)$$

It is also be useful to note that, for all \(I_1 \subset I_2\):
$$holes(I_1) \subset holes(I_2)$$
$$impls(I_1) \subset impls(I_2)$$
$$pars(I_1) \subset pars(I_2)$$
We also need to query certain syntactic elements of the code.
Some of these operations are not fully detailed in the interest of brevity.

We call locations or spans \emph{sites}.
We define a \texttt{HoleSite} as some data structure that captures the syntactic information of a hole called within code.

These are the operations on \texttt{HoleSite}s \(hs\), programs \(P\), implementations \(i\), and blocks \(b\):

\begin{description}
\item[{\(sites(i)\) or \(sites(P_{base})\)}] is the set of hole sites in the code of \(i\) or \(P_{base}\), so that \(|sites(c)| \geq |holes(c)|\).

\item[{\(sites(P)\)}] is the set of hole sites in all of the code of \(P\).

\item[{\(site(b)\)}] is the site of the start of the code of block \(b\).

\item[{\(block(hs)\)}] is the \texttt{block} that contains \(hs\), if any.

\item[{\(hole(hs)\)}] is the hole that is called at \(hs\).

\item[{\(scope(i)\)}] is the set of \texttt{block}s whose top-level declarations \(i\) references.

\item[{\(effects(i)\)}] \(\subset \texttt{Eff}\) is the set of \texttt{effect}s whose top-level declarations \(i\) references.
\end{description}

We can syntactically break down implementations \(i\) into a triple, \((i_{body}\), \(i_{return}\), \(i_{parameters})\), so that:
\begin{description}
\item[{\(i_{body}\)}] is the sequence of statements that makes of the code of the implementation,
\item[{\(i_{return}\)}] is the expression returned, if any,
\item[{\(i_{parameters}\)}] is the sequence of declarations of parameters made by the implementation.
\end{description}
\subsection{Structural constraints}
\label{sec:org0fd3dce}
    \label{structural-constraints}
Not all modular Stan programs that can be parsed are valid; we also impose certain structural and semantic constraints.
Input programs that do not meet the constraints are be rejected by the compiler.

Below are the structural constraints on a modular program \(P\):
\begin{enumerate}
\item \label{acyclic} The dependency graph of modules is acyclic:

For any graph \(G\), let \(N(G)\) be the set of nodes and \(E(G)\) be the set of edges.
Let the \emph{module dependency graph} \(MDG(P)\) be the directed graph with nodes \(N(MDG(P))=\texttt{impls}(P) \cup \{P_{base}\}\) and edges \(E(MDG(P)) = \st*{i_1 \rightarrow i_2}{i_1 \in impls(P) \cup \{P_{base}\}, i_2 \in impls(P), par(i_2) \in holes(i_1)}\).
We require that this graph is acyclic.
We also refer to this property as \(Acyclic(P)\).
\item \label{every-hole} Every hole has an implementation. \(\forall h \in holes(P) \cup holes(P_{base})\), \(impls(h) \neq \emptyset\).
\item \label{unique-ids} Hole identifiers are unique and (hole identifier, implementation identifier) pairs are unique.
\end{enumerate}

When a program \(P\) meets these constraints, we say \(valid_{structure}(P)\).
\subsection{Module signatures and semantic constraints}
\label{sec:orgdcf6368}
\subsubsection{Module Signatures}
\label{sec:orgc407e5b}
We would like to be able to guarantee that any concrete Stan program generated from a modular Stan program \(P\) is valid.
In order to do that, we need to understand the type, scope, and effect implications filling holes with their implementations.

To that end, we attempt to infer a \emph{signature} for every hole in \(P\).
A signature is like a function type plus extra information.
When no signature can be inferred for a hole, we reject the program as invalid.
Signatures let us specify semantic constraints on input programs, and are also be useful for generating Stan programs in \cref{concretize}.

A signature \(s\) is a tuple:
$$s = (s_{arg-types}, s_{ret-types}, s_{effects}, s_{scope})$$

\(s_{arg-types}\) specifies the argument types of a hole. \(s_{arg-types}\) is a sequence of Stan types, like the \texttt{TYPE} syntactic element introduced in \cref{stan-syntax}.

\(s_{ret-type}\) specifies the return type of a hole. \(s_{ret-type}\) is a Stan return type, like the \texttt{RET\_TYPE} syntactic element introduced in \cref{stan-syntax}.

\(s_{effects}\) refers to the set of effects that a hole's implementation may have.

\(s_{scope}\) refers to the scope of non-local variables that a hole's implementation may reference. 

\begin{enumerate}
\item Module Signature Inference
\label{sec:org5560ef2}
We define a procedure for either inferring the signature of a hole given the type determinations at each of the hole's call sites and all of the hole's implementations, or rejecting the program as invalid.

We rely on an operation \texttt{ReturnType} that infers the type of the expression returned by implementation code, if any, given that the types of all other expressions in that code are available.
Type annotating is a standard operation for compilers of typed languages including the Stan compiler \cite{stanc3-semantic-checking}.
We use the following interface:
$$\texttt{ReturnType}(arguments, i)$$
where \(arguments\) is the collection of variable names and types available within the code of the implementation \(i\).

We visit each hole in topological order by dependency, such that when we visit a hole \(h\), for all \(i \in impls(h)\), all of the holes \(h' \in holes(i)\) have already been visited.
This ordering is always possible because of the \(Acyclic(P)\) property.

\begin{align*}
  signature(h) = \big(&argtypes(impls(h)_0),\\
  &\texttt{ReturnType}(argtypes(impls(h)_0), impls(h)_0),\\
  &\bigcup_{i \in impls(h)} effects(i_{body}) \cup \bigcup_{h' \in holes(i)} signature(h')_{effects},\\
  &\bigcup_{i \in impls(h)} scope(i_{body}) \cup \bigcup_{h' \in holes(i)} signature(h')_{scope}\big)
\end{align*}

\(impls(h)_0\) is an arbitrary element of \(impls(h)\); \(impls(h)\) are never empty by structural constraint \cref{every-hole}.

This way, \(s_{arg-types}\) is assigned to the argument types of any of its implementations, \(s_{ret-types}\) is assigned to the type that can be inferred from any of its implementations with the return types of dependent holes, and \(s_{effects}\) \(s_{scope}\) are the unions of the effects and scopes required by any of the hole's implementations or descendants.
\end{enumerate}
\subsubsection{Semantic constraints}
\label{sec:orgbfd5f55}
\label{semantic-constraints}
We give a set of semantic constraints on programs in terms of signatures:

\begin{enumerate}
\item Implementations match signature argtypes. \(\forall i \in impls(P)\), \(i_{argtypes} = signature(par(i))_{argtypes}\)
\item Implementations match signature rettype. \(\forall i \in impls(P)\), \(i_{rettype} = signature(par(i))_{rettype}\)
\item Effects and scopes of holes are available in their block or module signature. 
\(\forall i \in impls(P), \forall st \in sites(i)\):
\begin{enumerate}
\item \(signature(st_{hole})_{effects} \subset signature(par(i))_{effects}\)
\item \(signature(st_{hole})_{scope} \subset signature(par(i))_{scope}\)
\end{enumerate}
\(\forall st \in sites(P_{base})\):
\begin{enumerate}
\item \(signature(st_{hole})_{effects} \subset effects(block)\)
\item \(signature(st_{hole})_{scope} \subset scope(block)\)
\end{enumerate}
\item Effects and scopes of code are available in their block or module signature.
\item \(P_{base}\) would typecheck under Stan if all holes \(h\) were function calls to function signatures with \(signature(h)_{argtypes}\) parameter types and \(signature(h)_{rettype}\) return type.
\item The body of each implementation \(i\) would typecheck under Stan if it were the body of a function with \(signature(par(i))_{argtypes}\) parameter types and \(signature(par(i))_{rettype}\) return type, and if the module-defined parameters \(i_{parameters}\) were included as model parameters.
\end{enumerate}

When a program \(P\) meets these constraints, we say \(valid_{semantics}(P)\).

Constraint 6 defines the scope available to code within modules: local variables, module arguments, module-defined global variables (such as parameters), and base-defined global variables (such as parameters).
\subsection{Program validity}
\label{sec:orgcce24bb}
We say that a modular Stan program \(P\) is valid if \(P\) meets both the structural constrains in \cref{structural-constraints} and the semantic constraints in \cref{semantic-constraints}: \(valid(P) = valid_{structure}(P) \land valid_{semantics}(P)\).
\subsection{Selections}
\label{sec:org002e592}
\emph{Selections} are subsets of the implementations in a modular program.
We call them ``selections'' because subsets of implementations are ``selected'' as components to build a concrete Stan model.
When a selection specifies a Stan program and has no extra implementations, the selection is \emph{valid}.

Next we give formal criteria to recognize valid selections, \(valid_P(I)\).
We show in \cref{concretize} that selections satisfying \(valid_P(I)\) are exactly those needed to define concrete Stan programs.
We first need to define \(siblings\), an intersection-like operation on selections.
\(siblings(I_1, I_2)\) is the set of pairs of implementations across two sets that are different but share a parent:
$$siblings(I_1, I_2) = \st*{(i_1, i_2)}{i_1 \in I_1, i_2 \in I_2, i_1 \neq i_2, par(i_1) = par(i_2)}$$

We note a useful property:
$$I_1 \subset I_2 \implies \forall I', siblings(I_1, I') \subset siblings(I_2, I')$$

\bigskip
The property \(valid_P(I)\) is true if and only if all of the following three criteria are met:
\begin{enumerate}
\item The selection only includes implementations that are in the program: $$I \subset impls(P)$$
\item Every hole in the program base and the selection has an implementation in the selection, and no extra implementations are included: $$impls(I) = holes(P_{base}) \cup holes(I)$$
\item The selection does not include any pair of implementations that implement the same hole, as these would be contradictory definitions: $$siblings(I, I) = \emptyset$$
\end{enumerate}

The following are convenient restatements of the above properties 2 and 3:
\begin{enumerate}
\item Each hole found in the program base and the selection has exactly one implementation in the selection: $$\forall h \in holes(P_{base}) \cup holes(I),\ |\st*{i}{i\in I, par(i) = h}| = 1$$
\item \(\forall h \in pars(I)\), \(h \in holes(I)\) or \(h \in holes(P_{base})\).
\end{enumerate}
\subsection{Semantic domain}
\label{sec:orgba8a0ea}
We are now equipped to formally define our high-level operations \cref{usecase:network,usecase:neighbors,usecase:concretize} and the semantic domain.

\(\texttt{Concretize}(P, I)\) is the concrete Stan program that results from including each implementation from the set valid selection \(I\) into the base of the valid modular program \(P\).
We consider an implementation of \texttt{Concretize} correct if for all modular programs \(P\) and selections \(I\) such that \(valid(P)\) and \(valid_P(I)\):

\begin{enumerate}
\item \(valid_{\texttt{STAN}}(\texttt{Concretize}(P, I))\)
\item \(\texttt{Concretize}(P, I)\) includes the same set of Stan statements in \(P_{base}\) and \(I\).
\end{enumerate}

\(\texttt{ModelGraph}(P)\) is the graph of all valid selections of \(P\), connected if they disagree on implementation of one hole.
\begin{align*}
  N(\texttt{ModelGraph}(P)) &= \st*{I}{I \subset impls(P), valid_P(I)}\\
  E(\texttt{ModelGraph}(P)) &= \st*{(I_1, I_2)}{I_1, I_2 \in N(\texttt{ModelGraph}(P)), |siblings(I_1, I_2)| = 1}
\end{align*}

\(\texttt{ModelNeighbors}(P, I)\) is the set of all selections that share an edge with \(I\) in the model graph:
   \(\texttt{ModelNeighbors}(P, I) = \st*{I'}{(I, I') \in E(\texttt{ModelGraph}(P))}\)

We can give a formal semantics of a valid modular program \(P\) in these terms: \(\semant{P}\) has the same graph structure as \(\texttt{ModelGraph(P)}\), but with the selection-valued nodes replaced by the corresponding probabilistic models:

\[N(\semant{P}) = \st*{\semant{\texttt{Concretize(P, I)}}}{I \in N(\texttt{ModelGraph}(P))}\]

In this way, a valid modular program \(P\) represents the graph of all probabilistic models that can be produced by recombination of modules, with connections between models that differ by only one choice of module.
\section{Algorithms}
\label{sec:org98db867}
\label{algorithms}
\subsection{\texttt{Concretize}}
\label{sec:orga63581d}
  \label{concretize}
In this section we develop a function \(\texttt{Concretize}(P, I)\): the concrete Stan program that is derived from a modular Stan program \(P\) by applying the implementations of a valid selection set \(I\).

We build \(\texttt{Concretize}\) by careful use of function inlining as a subroutine.
We use the following simplified interface:

\begin{verbatim}
InlineFunction(Program, CallSite, Stmts, Params)
InlineFunction(Program, CallSite, Stmts, Params, Return)
\end{verbatim}

Here, \texttt{Program} is the whole program to be updated, \texttt{CallSite} is a data structure indicating the span of code to be replaced, \texttt{Stmts} are a list of statements that make up the body of the function, \texttt{Params} is a list of the function's parameters, and \texttt{Return}, when present, is return expression (modules can only have zero or one \texttt{return} statements).
When we use \texttt{InlineFunction}, we take \texttt{Program} to be a modular Stan program, \texttt{CallSite} to be a \texttt{HoleSite}, \texttt{Stmts} to be \texttt{STMT\_FUNC\_M+}, \texttt{Params} to be \texttt{STMT\_DECL+}, and \texttt{Return} to be \texttt{EXPR\_M}.

\texttt{InlineFunction} is an operation standard in most optimizing compilers, including the Stan compiler \cite{opt-docs}.

\texttt{InlineFunction} assumes the following preconditions:
\begin{enumerate}
\item \texttt{Program} typechecks in the scope of \texttt{CallSite}.
\item \texttt{CallSite} supplies arguments matching \texttt{Params}.
\item The function represented by \texttt{Stmts}, \texttt{Params} and \texttt{Return} typechecks.
\end{enumerate}

Then we assume the following properties of \texttt{Program'} \texttt{=} \texttt{InlineFunction(Program,} \texttt{CallSite,} \texttt{Stmts,} \texttt{Params,} \texttt{Return)}:
\begin{enumerate}
\item References to \texttt{Params} in \texttt{Stmts} and \texttt{Return} are replaced by argument expression of \texttt{CallSite}.
\item \texttt{Stmts} are inserted in order before \texttt{CallSite} in the same scope as \texttt{CallSite}.
\item When \texttt{Return} is provided, \texttt{CallSite} is replaced by \texttt{Return}, which does not change the expression type.
\item \texttt{Program'} typechecks wherever \texttt{Program} typechecks.
\end{enumerate}

\subsubsection{\texttt{ApplyImpl}}
\label{sec:org1b65cc0}

We define an operation \(\texttt{ApplyImpl}(P, i)\), the result of taking a valid modular Stan program \(P\) and ``applying'' an implementation \(i \in impls(P)\), or using \(i\) to ``fill'' all instances of the hole \(par(i)\).
The hole that \(i\) fills, \(par(i)\), no longer appears in the resulting program.

\begin{verbatim}
ApplyImpl(P, i):
  sites = \st*{ site}{site \in HoleSites(p), site_{hole} = par(impl) }
  P' := P
  for site in sites:
    P' := InlineFunction(P', site, i_{body}, signature(st_{hole})_{arg-types}, i_{return})
  return InlineFunction(P', site(\texttt{parameters}), {}, i_{parameters})
\end{verbatim}

\noindent
\texttt{ApplyImpl} does not update module signatures.

\subsubsection{Correctness of \texttt{ApplyImpl}}
\label{sec:org2f5d003}
\begin{lemma}[$\texttt{ApplyImpl}(P, i)$ replaces the hole filled by $i$ with $i$'s holes]
\label{applyimpl-holes}
For all \(i' \in impls(P) \cup \{P_{base}\}\), \(holes_{\texttt{ApplyImpl}(P, i)}(i') = holes_P(i') \cup holes_P(i) - \{par(i)\}\) if \(par(i) \in holes_P(i')\) and \(holes_P(i')\) otherwise.
\end{lemma}
\begin{proof}
If \(par(i) \not\in holes_P(i')\), \texttt{ApplyImpl} does not make any replacements in \(i'\), so \(holes_{\texttt{ApplyImpl}(P, i)}(i') = holes_P(i')\).

Otherwise, \(\texttt{ApplyImpl}(P, i)\) is made up of code from \(i\) and \(i'\), so:
\begin{align*}
holes_{\texttt{ApplyImpl}(P, i)}(i') &\subset holes_P(i) \cup holes_P(i') \\
&= (holes_P(i) \cap \{par(i)\}) \cup (holes_P(i) - \{par(i)\}) \\
&\ \ \ \ \ \ \cup (holes_P(i') \cap \{par(i)\}) \cup (holes_P(i') - \{par(i)\})
\end{align*}

\(holes_P(i) \cap \{par(i)\} = \emptyset\), because otherwise \(i\rightarrow i \in E(MDG(P))\), which would violate \(Acyclic(P)\) and \(valid(P)\).

\texttt{ApplyImpl} replaces all \(h \in holes_P(i') \cap \{par(i)\}\) with \(i_{body}\), so \(h \not\in holes_{\texttt{ApplyImpl}(P, i)}(i')\).

Therefore:
$$holes_{\texttt{ApplyImpl}(P, i)}(i') \subset (holes_P(i) - \{par(i)\}) \cup (holes_P(i') - \{par(i)\})$$

\texttt{ApplyImpl} does not remove any \(holes_P(i') - \{par(i)\}\), so \(holes_P(i') - \{par(i)\} \subset holes_{\texttt{ApplyImpl}(P, i)}(i')\).

\texttt{ApplyImpl} inserts \(holes_P(i) - \{par(i)\}\) as part of \(i_{body}\), so \(holes_P(i) - \{par(i)\} \subset holes_{\texttt{ApplyImpl}(P, i)}(i')\).

Therefore:
\begin{align*}
  holes_{\texttt{ApplyImpl}(P, i)}(i') =& (holes_P(i) - \{par(i)\}) \cup (holes_P(i') - \{par(i)\})\\
  =& holes_P(i) \cup holes_P(i') - \{par(i)\}
\end{align*}
\end{proof}

\begin{lemma}[$\texttt{ApplyImpl}(P, i)$ replaces the hole filled by $i$ with $i$'s holes]
\label{applyimpl-effects}
For all sites \(st \in sites(P)\) such that \(st_{hole} = par(i)\), the code \(i_{body}\) replacing \(st\) in \(\texttt{ApplyImpl}(P, i)\) has \(effects(i_{body}) \subset signature(par(i))_{effects}\), \(scope(i_{body}) \subset signature(par(i))_{scope}\), and for all holes \(h \in holes(i)\), \(signature(h)_{effects}\) \(\subset\) \(signature(par(i))_{effects}\) and \(signature(h)_{scope}\) \(\subset\) \(signature(par(i))_{scope}\).
\end{lemma}
\begin{proof}
By definition of \(signature\):
$$signature(par(i))_{effects} \subset effects(i_{body}) \cup \bigcup_{h' \in holes(i)} signature(h')_{effects}$$
and
$$signature(par(i))_{scope} \subset scope(i_{body}) \cup \bigcup_{h' \in holes(i)} signature(h')_{scope}$$
\end{proof}

\begin{lemma}[Preconditions of ~InlineFunction~ are met]
\label{applyimpl-preconditions}
When \(valid(P)\), the preconditions of the \texttt{InlineFunction} calls in \texttt{ApplyImpl(P, i)} are met.
\end{lemma}
\begin{proof}
Consider the first call: $$InlineFunction(P, site, i_{body}, signature(st_{hole})_{arg-types}, i_{return})$$
\begin{enumerate}
\item If \(site\) is in the base of \(P\), then the \(P\) typechecks in the scope of \(site\) by semantic constraint 5, otherwise it must be in an implementation body, in which case it typechecks by semantic constraint 6.
\item We know \texttt{site} looks like a function call because of the definition of \texttt{HoleSites}. Its arguments must match \(signature(\texttt{site}_{hole})_{arg-types}\) by semantic constraint 5.
\item Implied by semantic constraint 6.
\end{enumerate}

\bigskip

On subsequent calls, \(P'\) is the output of a previous \texttt{InlineFunction}, which by induction meets its preconditions.
\(P'\) typechecks by the same argument as \(P\) because of property 4.
Both of the other preconditions hold by the same reasoning as the initial call.

\bigskip

Consider the final call: $$InlineFunction(P', site(\texttt{parameters}), \{\}, i_{parameters})$$
\begin{enumerate}
\item \(P'\) typechecks at \(site(\texttt{parameters})\) because \(P\) does by semantic constraint 5 and \(P'\) typechecks wherever \(P\) does.
\item \(site(\texttt{parameters})\) is a location rather than a hole site and so does not supply any arguments, which matches \texttt{\{\}}.
\item \texttt{Params} and \texttt{Return} are empty, and \texttt{Stmts} can only be a list of \texttt{STMT\_DECL} and so typechecks.
\end{enumerate}
\end{proof}

\begin{theorem}[$\texttt{ApplyImpl}$ produces valid modular Stan programs]
\label{applyimpl-correct}
\(valid(P), i \in P \implies valid(\texttt{ApplyImpl}(P, i))\)
\end{theorem}
\begin{proof}
Structural constraints:

\begin{enumerate}
\item \emph{The dependency graph of modules is acyclic}.

Suppose there is a cycle \(C\) in \(MDG(\texttt{ApplyImpl}(P, i)))\).

Let \(i_0 \rightarrow i_1\) be an edge in \(C\); then \(par(i_1) \in holes_{\texttt{ApplyImpl}(P, i)}(i_0)\).
By \cref{applyimpl-holes}, either \(par(i_1) \in holes_{P}(i_0)\), in which case \(i_0 \rightarrow i_1 \in E(MDG(P))\); or \(par(i) \in holes(i_0)\) and \(par(i_1) \in holes_{P}(i_0)\), in which case \(i_0 \rightarrow i, i \rightarrow i_1 \in E(MDG(P))\).

In either case, for all edges \(i_0 rightarrow i_1 \in C\), there exists a path from \(i_0\) to \(i_1\) in \(MDG(P)\), so there exists a cycle in \(MDG(P)\).
As that would contradict \(valid(P)\), there is no cycle \(C\) in \(\texttt{ApplyImpl}(P, i)\).
\item \emph{Every hole has an implementation.}

\texttt{ApplyImpl} does not remove any implementations from or add any new holes to \(P\), so any holes without implementations would also be without implementations in \(P\), which would violate \(valid(P)\).
\item \emph{Hole identifiers are unique and (hole identifier, implementation identifier) pairs are unique.}

\texttt{ApplyImpl} does not change any hole implementation identifier names, so any collisions would exist in \(P\) and violate \(valid(P)\).
\end{enumerate}

Semantic constraints:

\begin{enumerate}
\item \emph{Implementations match signature argtypes.}

\texttt{ApplyImpl} does not change signatures or implementation argument types.
\item \emph{Implementations match signature rettype}.
\(\forall i' \in impls(P)\), \(i'_{rettype} = signature(par(i'))_{rettype}\).

If \(i_{return}\) is a hole expression of \(par(i)\), then it is replaced by \(\texttt{ApplyImpl}(P, i)\).
By \cref{applyimpl-preconditions}, property 3 of \texttt{InlineFunction} for the replacement, so the type of \(i_{return}\) is not changed.

If the type of \(i_{return}\) depends on a hole expression of \(par(i)\), by the same reasoning, the type does not change. Otherwise, it is unaffected by \texttt{ApplyImpl}.
\item \emph{Effects and scopes of holes are available in their block or module signature}. 

Suppose \(\exists st \in sites(\texttt{ApplyImpl}(P, i))\) so that \(st\) violates one of the conditions.

If \(st \in sites(P)\), since \(block\) and \(signature\) are unaffected by \texttt{ApplyImpl}, then \(st\) would violate \(valid(P)\).
Therefore, \(st \not\in sites(P)\) and must have been part of a replacement of \(par(i)\) by \(i_{body}\).

By \cref{applyimpl-effects}, we have \(signature(st_{hole})_{effects} \subset signature(par(i))_{effects}\) and \(signature(st_{hole})_{scope} \subset signature(par(i))_{scope}\).
Since the scope and effects of \(st_{hole}\) are subsets of the scope and effects of \(par(i)\), and \(par(i)\) passes the four conditions, \(st_{hole}\) must also pass the four conditions.
\item \emph{Effects and scopes of code are available in their block or module signature.}

By identical reasoning to the previous case, for any code in a block or module implementation, the code either existed in \(P\), in which case it must be valid, or it replaced \(par(i)\), in which case by \cref{applyimpl-effects} it requires only a subset of the scope and effects of \(par(i)\), and is therefore also valid.
\item \emph{\(P_{base}\) would typecheck under Stan if all holes \(h\) were function calls to function signatures with \(signature(h)_{argtypes}\) parameter types and \(signature(h)_{rettype}\) return type.}

By \cref{applyimpl-preconditions}, property 4 of \texttt{InlineFunction} applies for all updates to \(P\) in \texttt{ApplyImpl}.
Since \(P_{base}\) typechecks, \(\texttt{ApplyImpl}(P, i)_{base}\) also typechecks.
\item \emph{The body of each implementation \(i\) would typecheck under Stan if it were the body of a function \(signature(par(i))_{argtypes}\) parameter types and \(signature(par(i))_{rettype}\) return type, and if the module-defined parameters \(i_{parameters}\) were included as model parameters.}

By \cref{applyimpl-preconditions}, property 4 of \texttt{InlineFunction} applies for all updates to  \(i_{body}\) for any \(i \in impls(P)\) in \texttt{ApplyImpl}.
Since \(i_{body}\) typechecks in \(P\) for all \(i \in impls(P)\), \(i_{body}\) also typechecks in \(\texttt{ApplyImpl}(P, i)\).
\end{enumerate}
\end{proof}

\subsubsection{\texttt{ApplyImpls}}
\label{sec:org0bb2ce0}

Let \(\texttt{ApplyImpls}(P, I)\) be the result of applying \(\texttt{ApplyImpl}(P, i)\) for each \(i \in I\).
Formally, if we let \(I = <i_1, \dots i_N>\), \(P_0 = P\), and \(P_j = ApplyImpl(P_{j-1}, i_j)\) for \(j \in [1, N]\), then \(\texttt{ApplyImpls}(P, I) = P_N\).

\subsubsection{Correctness of \texttt{ApplyImpls}}
\label{sec:org4701eb3}
\begin{theorem}[$\texttt{ApplyImpls}(P, I)$ adds $I$'s holes and removes the holes filled by $I$]
\label{applyimpls-holes}
For all \(i' \in impls(P) \cup \{P_{base}\}\), \(holes_{\texttt{ApplyImpls}(P, I)}(i') \subset holes_P(i') \cup holes_P(I) - pars(I)\)
\end{theorem}
\begin{proof}
Let \(I_j = <i_1, \dots i_j>\) so that \(I_N = I\).
Note that \(\texttt{ApplyImpls}(P, I_j) = P_j\).

We prove the following by induction on \(j\):
$$holes_{P_j}(i') \subset holes_P(i') \cup holes_P(I_j) - pars(I_j)$$

For \(j=0\):
\(holes_{P_0}(i') \subset holes_P(i') \cup holes_{P_0}(I) - pars(I) = holes_P(i') \cup \emptyset - \emptyset = holes_{P_0}(i')\).

For \(j>0\):
By \cref{applyimpl-holes}, \(holes_{P_j}(i') = holes_{\texttt{ApplyImpl}(P_{j-1}, i_j)}(i') \subset holes_{P_{j-1}}(i') \cup holes_{P_{j-1}}(i_j) - par(i_j)\).

By induction:

\(holes_{P_j}(i') \subset (holes_P(i') \cup holes_P(I_{j-1}) - pars(I_{j-1})) \cup holes_{P_{j-1}}(i_j) - par(i_j)\).

Again by induction, we can use \(holes_{P_{j-1}}(i_j) \subset holes_P(i_j) \cup holes_P(I_{j-1}) - pars(I_{j-1})\):

\(holes_{P_j}(i') \subset (holes_P(i') \cup holes_P(I_{j-1}) - pars(I_{j-1})) \cup (holes_P(i_j) \cup holes_P(I_{j-1}) - pars(I_{j-1})) - par(i_j)\).

\(holes_{P_j}(i') \subset holes_P(i') \cup (holes_P(I_{j-1}) \cup holes_P(i_j) \cup holes_P(I_{j-1})) - (pars(I_{j-1}) \cup par(i_j)\).

\(holes_{P_j}(i') \subset holes_P(i') \cup holes_P(I_j) - par(I_j)\).
\end{proof}
\subsubsection{\texttt{Concretize}}
\label{sec:orgd9e4596}
We define \texttt{Concretize} as \(\texttt{Concretize}(P, I) = \texttt{ApplyImpls}(P, I)_{base}\) when \(I\) is a valid selection.
\subsubsection{Correctness of \texttt{Concretize}}
\label{sec:org686d47b}
\begin{lemma}[$\texttt{Concretize}(P, I)$ has no holes]
\label{concretize-no-holes}
For \(valid(P)\) and \(valid_P(I)\), \(holes(\texttt{Concretize}(P, I)) = \emptyset\).
\end{lemma}
\begin{proof}
By \cref{applyimpls-holes}, \(holes(\texttt{ApplyImpls}(P, I)_{base}) \subset holes(P_{base}) \cup holes(I) - pars(I)\).
By \(valid_P(I)\), \(holes(P_{base}) \cup holes(I) = pars(I)\).
Therefore, \(holes(\texttt{ApplyImpls}(P, I)_{base}) \subset \emptyset\).
\end{proof}

\begin{theorem}[$\texttt{Concretize}(P, I)$ is a valid Stan program]
\label{concretize-correct}
\(valid(P)\) and \(valid_P(I)\) implies \(valid_{\texttt{Stan}}(\texttt{Concretize}(P, I))\).
\end{theorem}
\begin{proof}
By \cref{applyimpl-correct} we have \(valid(\texttt{ApplyImpls}(P, I))\).

By \cref{concretize-no-holes}, we have \(holes(\texttt{Concretize}(P, I)) = \emptyset\).
Therefore, \(SP = \texttt{Concretize}(P, I)\) is a concrete Stan program.

We can match the criteria of \(valid_{\texttt{Stan}}(SP)\) in \cref{valid-stan}:
\begin{enumerate}
\item \emph{Effects and scopes of code are available in their block}.
Implied by \(valid(\texttt{ApplyImpls}(P, I))\), semantic constraint 4.
\item \emph{\(SP\) typechecks as Stan code}.
Implied by \(valid(\texttt{ApplyImpls}(P, I))\), semantic constraint 5.
\end{enumerate}
\end{proof}

\subsubsection{User specification of a selection}
\label{sec:org68f3f98}
\label{selection-string}

In order for a user to call \texttt{Concretize}, they need to provide a ``selection''.

There is one more piece of syntax to define, though it does not appear in programs themselves: a \emph{selection string}.
A selection string is a string that references a set of module implementations within a given program.

\begin{verbatim}
SELECTION: hole_identifier:impl_identifier,+
\end{verbatim}

Selection strings can be used to narrow down and index into the set of programs represented by a modular program.
\subsection{\texttt{ModelGraph}}
\label{sec:orgd84950f}
In this section we present an efficient algorithm to implement the \(\texttt{ModelGraph}\) operation.
\subsubsection{Naive versions}
\label{sec:org6cc33dd}
To demonstrate what we mean by efficient, consider a naive construction of the model graph given a modular program \(P\):

   \begin{align*}
&N(\texttt{NaiveModelGraph}(P)) = \\
&\quad\st*{close(P_{base}, I)}{I \in \bigtimes_{h\in holes(p)} impls(h), valid_P(close(P_{base}, I))}\\
&E(\texttt{NaiveModelGraph}(P)) = \\
&\quad\st[\Big]{(I_1, I_2)}{I_1, I_2 \in N(\texttt{NaiveModelGraph}(P)), |siblings(I_1, I_2)| = 1}
   \end{align*}

Here, \(close(i, I)\) is the subset of \(I\) that is reachable from \(i\) by traversing edges in the module graph through \(I\).
Thus \(close(P_{base}, I)\) effectively removes redundant implementations from \(I\).
\texttt{NaiveModelGraph} works by enumerating all possible combinations of implementations and model pairs and then filtering out the invalid options.

\begin{figure}\centering
\begin{subfigure}{1.0\textwidth}
\includegraphics[width=\textwidth]{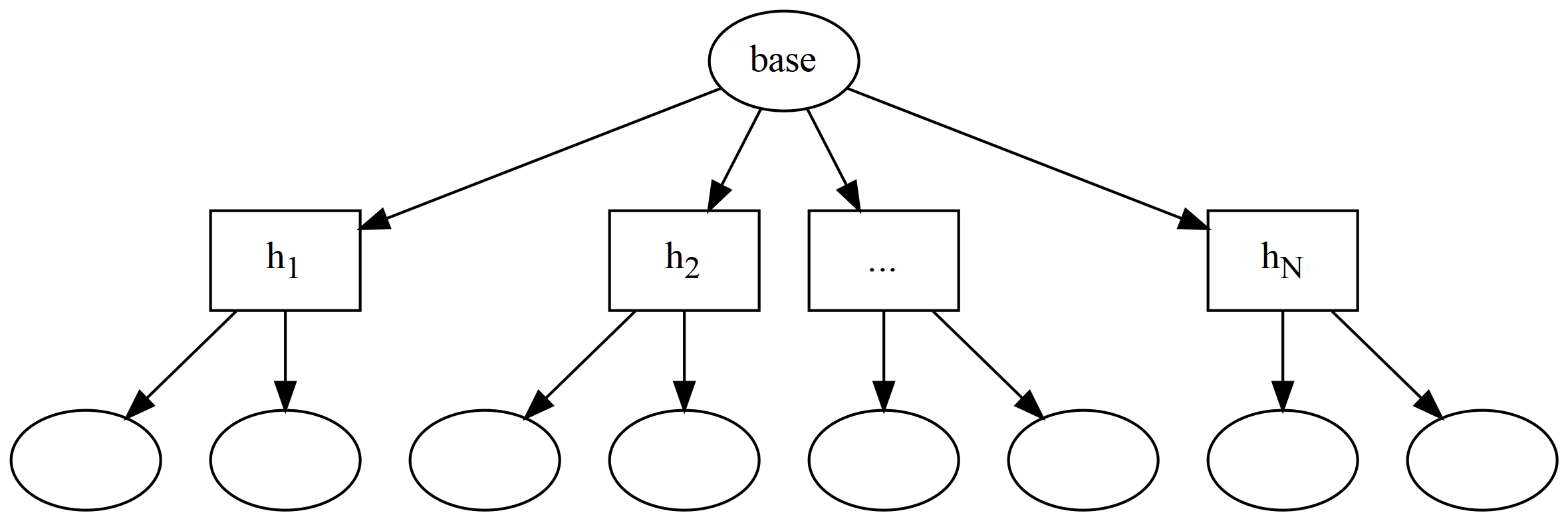}
\caption{}
\label{fig:patho:wide}
\end{subfigure}
\begin{subfigure}{0.3\textwidth}
\begin{center}
\includegraphics[width=\textwidth]{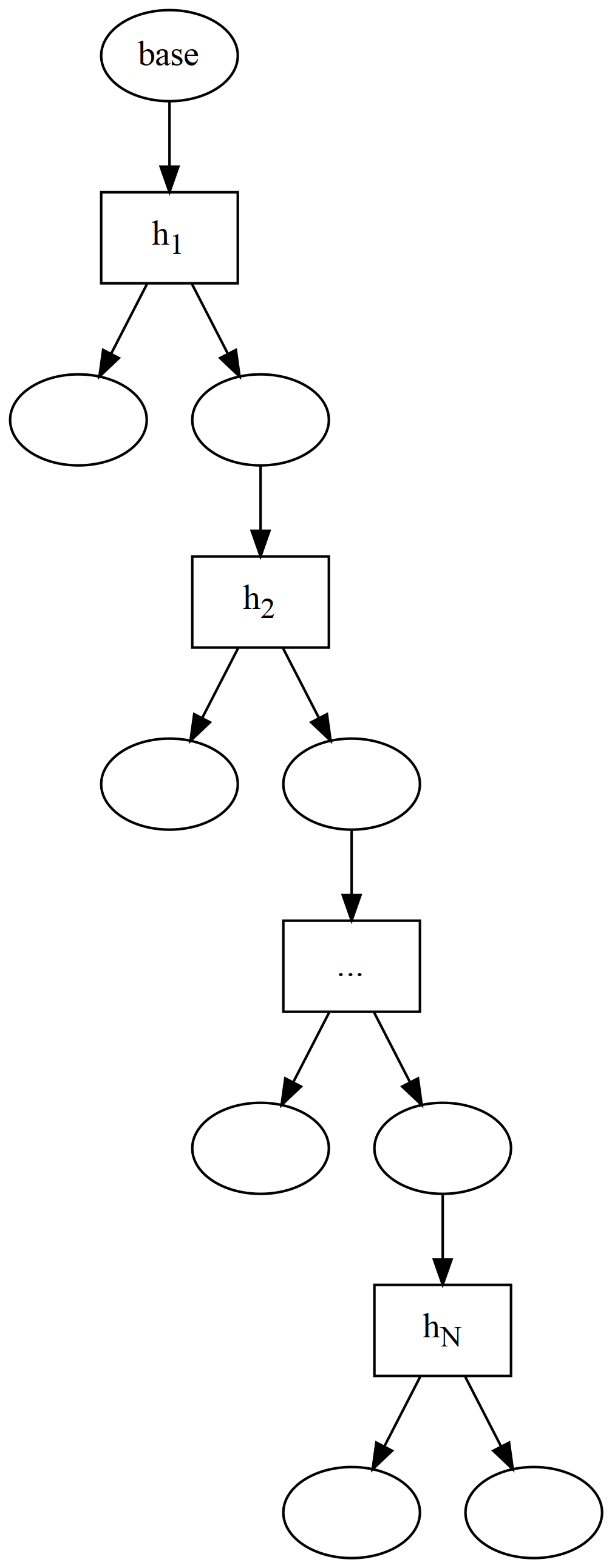}
\end{center}
\caption{}
\label{fig:patho:tall}
\end{subfigure}
\begin{subfigure}{0.3\textwidth}
\begin{center}
\includegraphics[width=\textwidth]{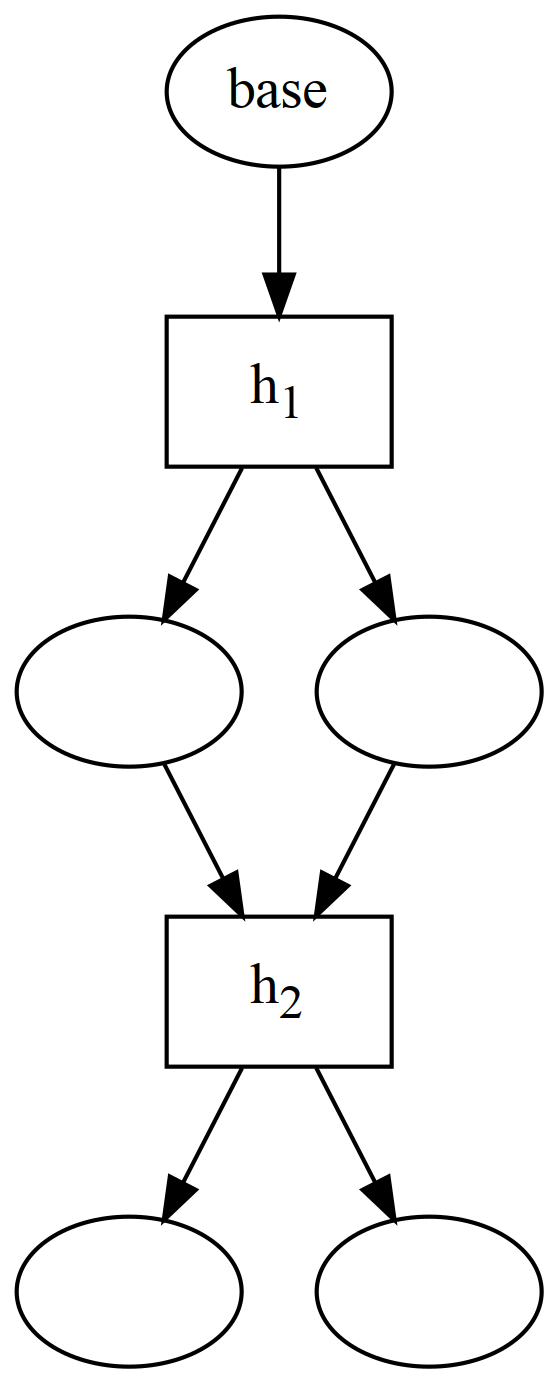}
\end{center}
\caption{}
\label{fig:patho:hole-branch}
\end{subfigure}
\begin{subfigure}{0.3\textwidth}
\begin{center}
\includegraphics[width=\textwidth]{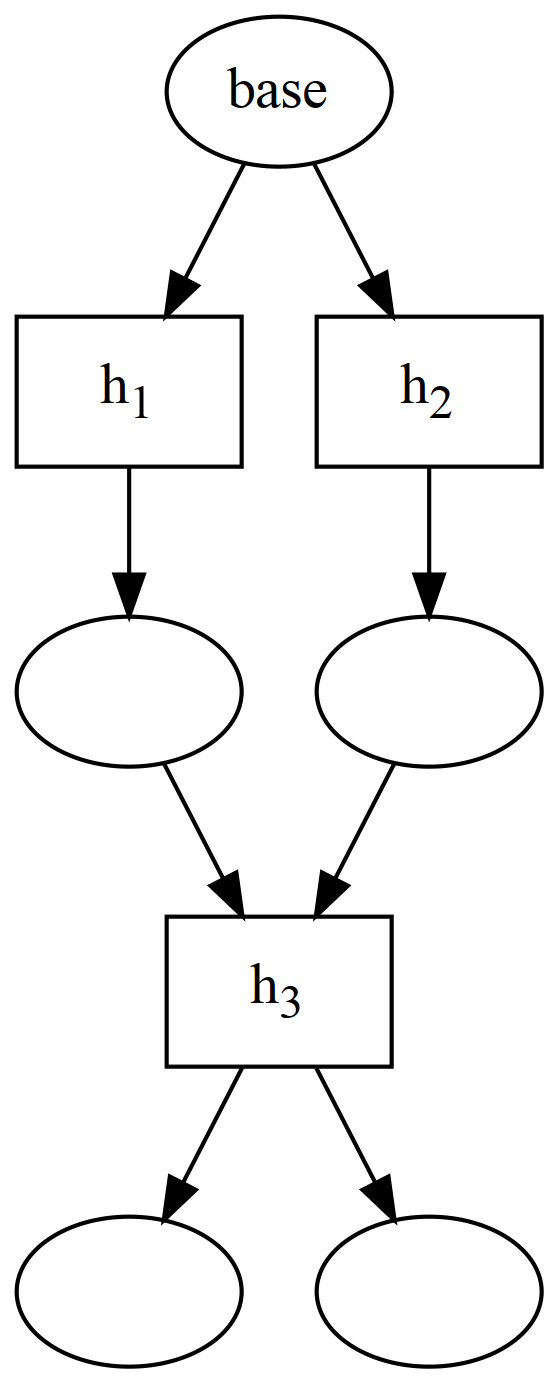}
\end{center}
\caption{}
\label{fig:patho:impl-branch}
\end{subfigure}
\caption{\label{fig:pathological}Pathological module graphs. Round nodes represent the base and module implementations while rectangular nodes represent holes.}
\end{figure}

\(E(\texttt{NaiveModelGraph}(P))\) is exponentially inefficient when \(P\) is \cref{fig:patho:wide}, because it would consider \(2^N\) candidate neighbors for each node, where in reality each node has only \(N-1\) neighbors.

\(N(\texttt{NaiveModelGraph}(P))\) is exponentially inefficient when \(P\) is \cref{fig:patho:tall}, because it would consider \(2^N\) candidate nodes, where in reality there are only \(N\) nodes.

Another approach would be a recursive construction. If the module graph is a tree, then the network of the sub-graph descending from each hole node is the graph join of the networks of the sub-graphs descending from its implementations, and the network of the sub-graph descending from each implementation node is the graph Cartesian product of the networks of the sub-graphs descending from its holes.
The issue with this approach is how to handle situations like \cref{fig:patho:hole-branch,fig:patho:impl-branch}: the module graph is not always a tree, so hole and implementation nodes will sometimes combine networks with overlapping information.
It is possible to modify the graph join and Cartesian product operations to handle the overlap correctly, but not without adding complexity and inefficiency that scales with the number recombination nodes.

\subsubsection{Algorithm}
\label{sec:orga0b0d62}
\label{modelgraph-algorithm}

Our algorithm builds up a the model graph's node and edge sets simultaneously by ``visiting'' each of the program's holes in turn.
The algorithm tracks the set of model and edge ``prefixes'': partial selections that only contain implementations for the holes visited thus far.
It also tracks the set of holes required so far by each prefix.
When the algorithm ``visits'' a hole, it ``expands'' the set of prefixes that require that hole to reflect its set of implementations.
To ensure that no ``prefix'' discovers that it requires a hole after that hole has been visited, and to avoid other complex edge cases, holes are visited in a topological order of dependency.

\begin{align*}
\texttt{expand}&(N, E, h) = \\
\Big( & \bigcup_{(H, I) \in N} \texttt{expand-node}(H, I, h), \\
& \bigcup_{(H, I) \in N} \texttt{new-edges}(H, I, h)\ \cup\
  \bigcup_{(H_1, I_1, H_2, I_2) \in E} \texttt{expand-edge}(H_1, I_1, H_2, I_2, h)
\Big)\\
\\
\texttt{expand}&\texttt{-node}(H, I, h) = \\
&\begin{dcases} 
\st*{\left( H \cup holes(i), I \cup \{i\} \right)}{i \in impls(h)}   & \text{if }h \in H \\
\left\{ \left( H, I \right) \right\}   & \text{otherwise} \\
\end{dcases}\\
\\
\texttt{new-ed}&\texttt{ges}(H, I, h) = \\
&\begin{dcases} 
\st*{\left( H \cup holes(i_1), I \cup \{i_1\}, H \cup holes(i_2), I \cup \{i_2\} \right)}{(i_1, i_2) \in {impls(h) \choose 2}}   & \text{if }h \in H \\
\emptyset   & \text{otherwise} \\
\end{dcases}\\
\\
\texttt{expand}&\texttt{-edge}(H_1, I_1, H_2, I_2, h) = \\
&\begin{dcases} 
\left\{ \left( H_1, I_1, H_2, I_2 \right) \right\}   & \text{if } h\not\in H_1, h\not\in H_2 \\
\st*{\left( H_1 \cup holes(i), I_1 \cup \{i\}, H_2, I_2 \right)}{i \in impls(h)}     & \text{if } h \in H_1, h\not\in H_2\\
\st*{\left( H_1 , I_1, H_2\cup holes(i), I_2 \cup \{i\} \right)}{i \in impls(h)}     & \text{if } h \not\in H_1, h\in H_2 \\
\st*{\left( H_1 \cup holes(i), I_1 \cup \{i\}, H_2\cup holes(i), I_2 \cup \{i\} \right)}{i \in impls(h)}     & \text{if } h \in H_1, h\in H_2 \\
\end{dcases}
\end{align*}

Let \(G_{Dep}(P)\) be the graph of dependencies between holes of a program \(P\) so that \(N(G_{Dep}(P)) = holes(P)\) and \(h_1 \rightarrow h_2 \in E(G_{Dep}(P))\) if and only if \(\exists i \in impls(h_1)\ s.t.\ h_2 \in holes(i)\).
Let \(\vec{H}\) be a topological ordering of \(G(P)\).

Let \((N_0, E_0) = \left( (holes(P_{base}), \emptyset), \emptyset \right)\).

Let \((N_j, E_j) = \texttt{expand}(N_{j-1}, E_{j-1}, h_j)\), \(\forall j \in [1, |\vec{H}|]\).

We show that \(\st{I}{(H, I)\in N_{|\vec{H}|}} = N(ModelGraph(P))\) and \(\st{(I_1, I_2)}{(H_1, I_1, H_2, I_2)\in E_{|\vec{H}|}} = E(ModelGraph(P))\).

\subsubsection{Proofs of correctness and completeness}
\label{sec:orgd620017}

Throughout this section, we assume that \(P\) is a modular program such that \(valid_{structure}(P)\).

\begin{lemma}[$\vec{H}$ is ordered by dependency]
\label{holes-topo}

There exists a topological ordering \(\vec{H} = ~topo~(G_{Dep}(P))\) and \(\forall j \in [1, |\vec{H}|], \forall i \in impls(P)\ s.t.\ \vec{H}_j \in holes(i), par(i) \in \vec{H}_{1:j-1} \text{ or } par(i) \in holes(P_{base})\).
\end{lemma}
\begin{proof}
The graph \(G_{Dep}\) must be acyclic, because otherwise the holes of \(P\) would have cyclic dependencies and \(P\) would not be valid.
Therefore \(G_{Dep}\) has at least one topological order \(\vec{H}\).

For any such \(j\) and \(i\), we must have \(par(i) \rightarrow \vec{H}_j \in E(G_{Dep}(P))\), by definition of \(G_{Dep}\).
Then either \(par(i) \in holes(P_{base})\) and we are done, or \(par(i) \in holes(P)\), and so \(par(i) \in \vec{H}\), so \(\vec{H}_k\) for some \(k\).
By the definition of a topological ordering, \(k < j\).
Therefore, \(par(i) \in \vec{H}_{1:j-1}\).
\end{proof}

\begin{lemma}[Model prefixes have complete prefix hole sets]
\label{holes-complete}
\(\forall j \in [1, |\vec{H}|], \forall (H_{j-1}, I_{j-1}) \in N_{j-1}\), let \(N^I = \st{I'}{I' \in N(ModelGraph(P)), I_{j-1} \subset I'}\).

For all \(I'\) in \(N^I\), \(pars(I') \cap \vec{H}_{1:j} = H_{j-1} \cap \vec{H}_{1:j}\).
\end{lemma}
\begin{proof}
This is a proof by induction on \(j\).

We will simultaneously prove the following property, which we refer to as ``consistency,'' by induction for each \(j\):
  $$I' \cap impls(\vec{H}_{1:j-1}) = I_{j-1} \cap impls(\vec{H}_{1:j-1})$$

For \(j=1\):

Since \(N_0 = \{\emptyset\}\), \(I_{j-1} = \emptyset\), so \(N^I = N(ModelGraph(P))\).

If \(\vec{H}_1 \in pars(I')\), then by \(valid_P(I')\), \(\vec{H}_1 \in holes(P_{base})\), since if \(\exists i\in impls(I')\ s.t.\ \vec{H}_1 \in holes(i)\), then by \cref{holes-topo}, \(par(i) \in \vec{H}_{1:0} = \emptyset\).
Then since \(holes(P_{base}) \subset pars(I')\) and \(H_0 = holes(P_{base})\), \(pars(I') \cap \vec{H}_{1:1} = H_0 \cap \vec{H}_{1:1}\).

If \(\vec{H}_1 \not\in pars(I')\), then \(\vec{H}_1 \not\in holes(P_{base}) = H_0\), so \(pars(I') \cap \vec{H}_{1:1} = H_0 \cap \vec{H}_{1:1} = \emptyset\).
Therefore the lemma holds for \(j=1\).

Trivially, consistency holds for \(j=1\) because \(\not\exists h \in \vec{H}_{1:0}\).

For \(j>1\):

Since \((H_{j-1}, I_{j-1}) \in N_{j-1}\) and \(j \geq 2\), \(\exists (H_{j-2}, I_{j-2}) \in N_{j-2}\) such that \((H_{j-1}, I_{j-1}) \in \texttt{expand-node}(H_{j-2}, I_{j-2}, \vec{H}_{j-1})\).

If \(\vec{H}_{j-1} \in H_{j-2}\), then by induction of the lemma, \(\vec{H}_{j-1} \in pars(I')\), so \(I' \cap impls(\vec{H}_{j-1}) = \{i'\}\) for some \(i'\).
It also follows from the definition of \(\texttt{expand-node}\) that, for some \(i \in impls(\vec{H}_{j-1})\), \(I_{j-1} = I_{j-2} \cup \{i\}\), and since \(I_{j-2} \cap impls(\vec{H}_{j-1}) = \emptyset\), \(I_{j-1} \cap impls(\vec{H}_{j-1}) = \{i\}\).
Since \(I_{j-1} \subset I'\), \(i \in I'\).

If \(i/neq i'\), then since \(par(i) = par(i')\), \((i, i') \in siblings(I', I')\), but \(valid_{P}(I')\) implies \(siblings(I', I') = \emptyset\), so \(i = i'\).

Thus, if \(\vec{H}_{j-1} \in H_{j-2}\), then \(I' \cap impls(\vec{H}_{j-1}) = I_{j-1} \cap impls(\vec{H}_{j-1}) = \{i\}\).
If \(\vec{H}_{j-1} \not\in H_{j-2}\), then \(\vec{H}_{j-1} \not\in pars(I')\) and \(I_{j-1} = I_{j-2}\), so \(I' \cap impls(\vec{H}_{j-1}) = I_{j-1} \cap impls(\vec{H}_{j-1}) = \emptyset\).

Starting from induction on consistency:
$$I' \cap impls(\vec{H}_{1:j-2}) = I_{j-2} \cap impls(\vec{H}_{1:j-2})$$
$$(I' \cap impls(\vec{H}_{1:j-2})) \cup (I' \cap impls(\vec{H}_{j-1})) = (I_{j-2} \cap impls(\vec{H}_{1:j-2})) \cup (I_{j-1} \cup impls(\vec{H}_{j-1}))$$
Because \(I_{j-2} = I_{j-1} \cap impls(\vec{H}_{1:j-2})\):
$$(I' \cap impls(\vec{H}_{1:j-2})) \cup (I' \cap impls(\vec{H}_{j-1})) = (I_{j-1} \cap impls(\vec{H}_{1:j-2})) \cup (I_{j-1} \cup impls(\vec{H}_{j-1}))$$
$$I' \cap impls(\vec{H}_{1:j-1}) = I_{j-1} \cap impls(\vec{H}_{1:j-1})$$

Thus, consistency holds.

If \(\vec{H}_j \in holes(P_{base})\), then the lemma is true by the same reasoning as the \(j=1\) case.

Suppose \(\vec{H}_j \in H_{j-1}\) and \(\vec{H}_j \not\in holes(P_{base})\).
By the construction of \(H_{j-1}\) and the definition of \(\texttt{expand-node}\), \(\exists i \in I_{j-1}\) such that \(\vec{H}_j \in holes(i)\).
\(I_{j-1} \subset I'\), \(i \in I'\), so \(\vec{H}_j = par(i) \in pars(I')\).

Suppose \(\vec{H}_j \in pars(I')\) and \(\vec{H}_j \not\in holes(P_{base})\).
By \(valid_P(I')\) and \cref{holes-topo}, there exists some \(k < j\) such that \(par(i) = \vec{H}_k\) and \(i \in impls(\vec{H}_k)\).
By consistency, \(I' \cap impls(\vec{H}_{1:k}) = I_k \cap impls(\vec{H}_{1:k})\), so \(i \in I_k \subset I_{j-1}\).
By the definition of \(\texttt{expand-node}(H_{k-1}, I_{k-1}, \vec{H}_k)\), \(i \in I_k \implies holes(i) \subset H_k\), so \(\vec{H}_j \in H_k \subset H_{j-1}\).
\end{proof}

The following lemma shows that a model prefix can be expanded with any implementation on its horizon and remain a valid model prefix.

\begin{lemma}[Model prefixes can be expanded with any implementation]
\label{any-impl}
\(\forall j \in [1, |\vec{H}|], \forall I \in N(ModelGraph(P)), \forall h \in holes(I \cap impls(\vec{H}_{1:j}) \cup holes(P_{base})) - \vec{H}_{1:j}, \forall i \in impls(h), \exists I' \in N(ModelGraph(P))\ s.t.\ I \cap impls(\vec{H}_{1:j}) \cup \{i\} \subset I'\).
\end{lemma}
\begin{proof}
We will construct a selection set that satisfies the property by modifying \(I\).

Let \(I_1 = I - impls(h) \cup \{i\}\).
Let \(I_2 = I_1 \cup \st{\texttt{any}(impls(h))}{h \in holes(P) - pars(I_2)}\), where \(\texttt{any}(s)\) is an arbitrary element of a set.
The sets \(impls(h)\) are never empty by \(valid(P)\).
Then, \(pars(I_2) = holes(P)\) and \(I \cap impls(\vec{H}_{1:j}) \cup \{i\} \subset I_2\).

Let \(f(I) = \st{i}{i \in I, par(i) \in holes(I) \cup holes(P_{base})}\).
Let \(I_3\) be the fixed point applying \(f\) to \(I_2\).
We know \(f\) converges to some \(I_3\) because \(|I|\) is finite, cannot go below \(0\), and decreases monotonically under \(f\) until convergence.

\(I_3\) meets the criteria for \(valid_P(I_3)\):

\begin{enumerate}
\item \(I_3 \subset impls(P)\): \(I_3 \subset I_2 \subset impls(holes(P)) \cup I \cup \{i\} \subset impls(P)\) by \(valid_P(I)\).
\item By \(I_3 = f(I_3)\) we know \(pars(I_3) \subset holes(I_3) \cup holes(P_{base})\).
Since \(pars(I_2) = holes(P) \supset holes(I_3) \cup holes(P_{base})\) and \(f\) does not remove any \(h \in holes(I_3) \cup holes(P_{base})\), we have \(pars(I_3) = holes(I_3) \cup holes(P_{base})\).
\item \(siblings(I_3, I_3) = \emptyset\), because only one \(i \in impls(h)\) for each \(h \not\in pars(I)\) is added to \(I\) in the construction of \(I_1\) and \(I_2\).
\end{enumerate}

By construction, \(I \cap impls(\vec{H}_{1:j}) \cup \{i\} \subset I_3\).
\end{proof}

\begin{lemma}[$N_j$ represents exactly the model prefixes]
\label{nodes-subsets}

\(\forall j \in [0, |\vec{H}|], \st{I}{(H, I) \in N_j} = \st{I \cap impls(\vec{H}_{1:j})}{I \in N(ModelGraph(P))}\).
\end{lemma}
\begin{proof}
This is a proof by induction on \(j\).

For \(j=0\), 
\(\st{I}{(H, I) \in N_0} = \{\emptyset\} = \st{I \cap \emptyset}{I \in N(ModelGraph(P))}\).

For \(j \geq 1\):

Suppose \(\exists I \in \st{I \cap impls(\vec{H}_{1:j})}{I \in N(ModelGraph(P))}\) such that \(I \not\in \st{I}{(H, I) \in N_j}\).
Let \(I' = I \cap impls(\vec{H}_{1:j-1})\).
By induction, \(I' \in \st{I}{(H, I) \in N_{j-1}}\), so \(\exists H'\ s.t.\ (H', I') \in N_{j-1}\).

If \(\vec{H}_j \in H'\), then by \cref{holes-complete}, \(\vec{H}_j \in pars(I)\), so \(\exists i \in I\ s.t.\ i \in impls(\vec{H}_j)\).
By the definition of \(\texttt{expand-node}(H', I', \vec{H}_j)\), \(\forall i' \in impls(\vec{H}_j), I' \cup \{i'\} \in \st{I}{(H, I) \in N_j}\), so \(I' \cup \{i\} = I \in \st{I}{(H, I) \in N_j}\), which contradicts the construction of \(I\).

If \(\vec{H}_j \not\in H'\), then by \cref{holes-complete}, \(\vec{H}_j \not\in pars(I)\), so \(I' = I\).
By the definition of \(\texttt{expand-node}(H', I', \vec{H}_j)\), \((H', I') \in N_j\), so \(I \in \st{I}{(H, I) \in N_j}\), which again contradicts the construction of \(I\).

Therefore, there is no such I, and \(\st{I \cap impls(\vec{H}_{1:j})}{I \in N(ModelGraph(P))} \subset \st{I}{(H, I) \in N_j}\).

Suppose \(\exists I \in \st{I}{(H, I) \in N_j}\) such that \(I \not\in \st{I \cap impls(\vec{H}_{1:j})}{I \in N(ModelGraph(P))}\).
By the definition of \(N_j\), \(\exists (H', I')\in N_{j-1}\ s.t.\ (H, I) \in \texttt{expand-node}(H', I', \vec{H}_{j})\).
By induction, \(I' \in \st{I \cap impls(\vec{H}_{1:j-1})}{I \in N(ModelGraph(P))}\), and therefore \(\exists I^* \in N(ModelGraph(P))\ s.t.\ I' \subset I^*\).

If \(\vec{H}_{j} \not\in H'\), then by the definition of \texttt{expand-node}, \(I = I'\), and by \cref{holes-complete}, \(\vec{H}_j \not\in pars(I^*)\); so \(I = I' = I^* \cap impls(\vec{H}_{1:j-1}) = I^* \cap impls(\vec{H}_{1:j})\), so \(I \in \st{I \cap impls(\vec{H}_{1:j})}{I \in N(ModelGraph(P))}\).

If \(\vec{H}_{j} \in H'\), then by the definition of \texttt{expand-node}, \(\exists i \in impls(\vec{H}_j)\ s.t.\ I = I' \cup \{i\}\), and by \cref{holes-complete}, \(\vec{H}_j \in pars(I^*)\).

Because \(\vec{H}_j \in holes(I^* \cap impls(\vec{H}_{1:j-1})) - \vec{H}_{1:j-1}\), by \cref{any-impl}, \(\exists I^i \in N(ModelGraph(P))\ s.t.\ I \cap impls(H_{1:j-1}) \cup \{i\} \subset I^i\), so \(I' \cup \{i\} = I \subset I^i\), so \(I \in \st{I \cap impls(\vec{H}_{1:j})}{I \in N(ModelGraph(P))}\).

Therefore, \(\st{I}{(H, I) \in N_j} \subset \st{I \cap impls(\vec{H}_{1:j})}{I \in N(ModelGraph(P))}\).
\end{proof}

\begin{theorem}[$N_{|\vec{H}|}$ is correct and complete]
\label{nodes-correct}
\(\st{I}{(H, I) \in N_{|\vec{H}|}} = N(ModelGraph(P))\).
\end{theorem}
\begin{proof}
By \cref{nodes-subsets}, \(\st{I}{(H, I) \in N_{|\vec{H}|}} = \st{I \cap impls(\vec{H}_{1:|\vec{H}|})}{I \in N(ModelGraph(P))}\), and because \(\vec{H}_{1:|\vec{H}|} = pars(P)\), \(\st{I \cap impls(\vec{H}_{1:|\vec{H}|})}{I \in N(ModelGraph(P))} = N(ModelGraph(P))\).
\end{proof}

\begin{lemma}[Edge prefix endpoints are model prefixes]
\label{edge-endpoints}
\(\forall j \in [0, |\vec{H}|]\), \(\forall (H_1\), \(I_1\), \(H_2\), \(I_2) \in E_j\), \((H_1, I_1) \in N_j\) and \((H_2, I_2) \in N_j\).
\end{lemma}
\begin{proof}
This is a proof by induction on \(j\).

For \(j=0\), the lemma is trivially satisfied because \(E_0 = \emptyset\).

For \(j\geq 1\):

By definition of \texttt{expand}, \((H_1, I_1, H_2, I_2)\) came either from \texttt{new-edges} \texttt{expand-edge}.

If \((H_1, I_1, H_2, I_2) \in \texttt{new-edges}(H', I', \vec{H}_j)\) for some \((H', I') \in N_{j-1}\), then \(\vec{H}_j \in H'\) and \((H_1, I_1) = (H' \cup holes(i_1), I' \cup \{i_1\})\), \((H_2, I_2) = (H' \cup holes(i_2), I' \cup \{i_2\})\), so \((H_1, I_1), (H_2, I_2) \in \texttt{expand-node}(H', I', \vec{H}_j) \subset N_j\).

If \((H_1, I_1, H_2, I_2) \in \texttt{expand-edge}(H_1', I_1', H_2', I_2', \vec{H}_j)\) for some \((H_1', I_1', H_2', I_2') \in E_{j-1}\), then by induction \((H_1', I_1'), (H_2', I_2') \in N_{j-1}\).
If \(\vec{H}_j \in H_1'\), then \(\exists i \in impls(H_1')\ s.t.\ (H_1, I_1) = (H_1' \cup holes(i_1), I_1' \cup \{i_1\}) \in \texttt{expand-node}(H_1', I_1', \vec{H}_j) \subset N_j\).
If \(\vec{H}_j \not\in H_1'\), then \((H_1, I_1) = (H_1', I_1') \in \texttt{expand-node}(H_1', I_1', \vec{H}_j) \subset N_j\).
The same reasoning applies for \((H_2, I_2)\).
\end{proof}

\begin{lemma}[Edge prefix endpoints share one pair of siblings]
\label{edge-neighbors}
\(\forall j \in [0, |\vec{H}|], \forall (H_1, I_1, H_2, I_2) \in E_j\), \(|siblings(I_1, I_2)| = 1\).
\end{lemma}
\begin{proof}
This is a proof by induction on \(j\).

For \(j=0\), the lemma is trivially satisfied because \(E_0 = \emptyset\).

For \(j\geq 1\):

By definition of \texttt{expand}, \((H_1, I_1, H_2, I_2)\) came either from \texttt{new-edges} \texttt{expand-edge}.

If \((H_1, I_1, H_2, I_2) \in \texttt{new-edges}(H', I', \vec{H}_j)\) for some \((H', I') \in N_{j-1}\), then \(\exists (i_1, i_2) \in {impls(\vec{H}_j) \choose 2}\) such that \(I_1 = I' \cup \{i_1\}\), \(I_2 = I' \cup \{i_2\}\), \(par(i_1) = par(i_2) = \vec{H}_j\), so \((i_1, i_2) \in siblings(I_1, I_2)\).
Since \(I_1\) and \(I_2\) are otherwise identical, \(siblings(I_1, I_2)\) can have no other elements.

If \((H_1, I_1, H_2, I_2) \in \texttt{expand-edge}(H_1', I_1', H_2', I_2', \vec{H}_j)\) for some \((H_1', I_1', H_2', I_2') \in E_{j-1}\), then by induction \(siblings(I_1', I_2') = \{(i_1, i_2)\}\) for some \(i_1, i_2\).
Since \((I_1, I_2)\) is either \((I_1', I_2')\) or \((I_1' \cup \{i\}, I_2' \cup \{i\})\) for some \(i \in impls(\vec{H}_j)\), \(I_1\) and \(I_2\) cannot have gained or lost any siblings, so \(siblings(I_1, I_2) = \{(i_1, i_2)\}\).
\end{proof}

\begin{lemma}[Model prefixes without siblings are equal]
\label{no-siblings-equal}
\(\forall I_1, I_2 \in N(ModelGraph(P))\), \(\forall j \in [1, |\vec{H}|]\), let \(I_1' = I_1 \cap impls(\vec{H}_{1:j})\) and \(I_2' = I_2 \cap impls(\vec{H}_{1:j})\). If \(siblings(I_1, I_2) = \emptyset\), then \(I_1' = I_2'\).
\end{lemma}
\begin{proof}
Let \(i\) be the element of \(I_1' - I_2'\) so that \(par(i) = \vec{H}_k\) with minimal \(k\).
Since \(i\not\in I_2'\) and \(siblings(I_1', I_2') = \emptyset\), \(\vec{H}_k \not\in pars(I_2')\).
Since \(valid_P(I_1)\) and \(\vec{H}_k \in pars(I_1')\), \(\vec{H}_k \in holes(P_{base}) \cup holes(I')\).
Since \(valid_P(I_2)\), \(holes(P_{base}) \subset pars(I_2')\), so \(\vec{H}_k \not\in holes(P_{base})\), so \(\exists i' \in I_1'\ s.t.\ \vec{H}_j \in holes(i')\)
By the topological sorting of \(\vec{H}\), \(par(i') = \vec{H}_m\) for some \(m < k\).
Since \(k\) is minimal in \(I_1' - I_2'\), \(i' \not\in I_1'-I_2'\), so \(i' \in I_2'\).
Then \(\vec{H}_k \in holes(i') \subset holes(I_2') \subset pars(I_2')\), but \(\vec{H}_k \not\in pars(I_2')\).
Therefore \(I_1' - I_2'\) is empty.
The same reasoning applies to show \(I_2' - I_1'\) is empty, so \(I_1' = I_2'\).
\end{proof}

\begin{lemma}[$E_j$ is a complete set of edge prefixes]
\label{edges-complete}
\(\forall j \in [1, |\vec{H}|], \st{\big(I_1 \cap impls(\vec{H}_{1:j}), I_2 \cap impls(\vec{H}_{1:j})\big)}{I_1, I_2 \in N(ModelGraph(P)), siblings(I_1, I_2) = \{(i_1, i_2)\}, par(i_1) \in \vec{H}_{1:j}} \subset \st{(I_1, I_2)}{(H_1, I_1, H_2, I_2) \in E_j}\).
\end{lemma}
\begin{proof}
This is a proof by induction on \(j\).

For \(j=0\), 
\(\st{\big(I_1 \cap \emptyset, I_2 \cap \emptyset\big)}{I_1, I_2 \in N(ModelGraph(P)), siblings(I_1, I_2) = \{(i_1, i_2)\}, par(i_1) \in \emptyset} = \emptyset \subset E_0 = \emptyset\).

For \(j \geq 1\):

We consider some \((I_1, I_2) \in \st{\big(I_1 \cap impls(\vec{H}_{1:j}), I_2 \cap impls(\vec{H}_{1:j})\big)}{I_1, I_2 \in N(ModelGraph(P)), siblings(I_1, I_2) = \{(i_1, i_2)\}, par(i_1) \in \vec{H}_{1:j}}\) and show that \((I_1, I_2) \in \st{(I_1, I_2)}{(H_1, I_1, H_2, I_2) \in E_j}\).
Let \(\{(i_1, i_2)\} = siblings(I_1, I_2)\).

If \(par(i_1) = par(i_2) = \vec{H}_{j}\), then since \(siblings(I_1 - i_1, I_2 - i_2) = \emptyset\), by \cref{no-siblings-equal}, \(I_1 - i_1 = I_2 - i_2\).
Since \(I_1 - i_1 \subset impls(\vec{H}_{1:j-1})\), by \cref{nodes-subsets}, \(\exists H'\ s.t.\  (H', I_1 - i_1) \in N_{j-1}\), therefore \(\texttt{new-edges}(H', I_1 - i_1, \vec{H}_j) \subset E_j\), and because \(\vec{H}_j \in pars(I_1)\) implies \(\vec{H}_j \in H'\) and \((i_1, i_2) \in {impls(\vec{H}_j) \choose 2}\), \((H' \cup holes(i_1), I_1 - i_1 \cup i_1 = I_1, H' \cup holes(i_2), I_1 - i_1 \cup i_2 = I_2) \in \texttt{new-edges}(H', I_1 - i_1, \vec{H}_j)\).

If \(par(i_1) = par(i_2) \neq \vec{H}_{j}\), then by induction, \((I_1 \cap impls(\vec{H}_{1:j-1}), I_2 \cap impls(\vec{H}_{1:j-1})) \in \st{(I_1, I_2)}{(H_1, I_1, H_2, I_2) \in E_{j-1}}\), so \(\exists (H_1', I_1', H_2', I_2') \in E_{j-1}\) such that \(I_1' = I_1 \cap impls(\vec{H}_{1:j-1})\) and \(I_2' = I_2 \cap impls(\vec{H}_{1:j-1}))\).

If \(\vec{H}_j \not\in H_1'\) and \(\vec{H}_j \not\in H_2'\), then by \cref{edge-endpoints} and \cref{holes-complete}, \(I_1 = I_1'\), \(I_2 = I_2'\), and by the definition of \(\texttt{expand-edge}\), \((H_1', I_1', H_2', I_2') = (H_1', I_1, H_2', I_2) \in \texttt{expand-edge}(H_1', I_1', H_2', I_2', \vec{H}_j) \subset E_j\).

If \(\vec{H}_j \in H_1'\) and \(\vec{H}_j \not\in H_2'\), then \(\exists i \in impls(\vec{H}_j)\) such that \(I_1 = I_1' \cup \{i\}\) and \(I_2 = I_2'\), and \(\forall i' \in impls(\vec{H}_j)\), \((H_1' \cup holes(i'), I_1' \cup \{i'\}, H_2', I_2') = (H_1' \cup holes(i'), I_1, H_2', I_2) \in \texttt{expand-edge}(H_1', I_1', H_2', I_2', \vec{H}_j) \subset E_j\).
Similarly, if \(\vec{H}_j \not\in H_1'\) and \(\vec{H}_j \in H_2'\), then \(I_1 = I_1', I_2 = I_2' \cup \{i\}\) and \((H_1', I_1', H_2' \cup holes(i), I_2' \cup \{i\}) = (H_1', I_1, H_2' \cup holes(i), I_2) \in \texttt{expand-edge}(H_1', I_1', H_2', I_2', \vec{H}_j) \subset E_j\).

If \(\vec{H}_j \in H_1'\) and \(\vec{H}_j \in H_2'\), then since \(par(i_1) \neq \vec{H}_j\), \(\exists i \in impls(\vec{H}_j)\) such that \(I_1 = I_1' \cup \{i\}\) and \(I_2 = I_2' \cup \{i\}\), and \(\forall i' \in impls(\vec{H}_j)\), \((H_1' \cup holes(i), I_1' \cup \{i\}, H_2' \cup holes(i), I_2' \cup \{i\}) = (H_1' \cup holes(i), I_1 \cup \{i\}, H_2' \cup holes(i), I_2) \in \texttt{expand-edge}(H_1', I_1', H_2', I_2', \vec{H}_j) \subset E_j\).
\end{proof}

\begin{theorem}[$E_{|\vec{H}|}$ is the correct and complete edge set]
\label{edges-correct}
\(\st{(I_1, I_2)}{(H_1, I_1, H_2, I_2) \in E_{|\vec{H}|}} = E(ModelGraph(P))\).
\end{theorem}
\begin{proof}
\(\forall (H_1, I_1, H_2, I_2) \in E_{|\vec{H}|}\), \(I_1, I_2 \in N_{|\vec{H}|}\) by \cref{edge-endpoints}, so \(I_1, I_2 \in N(ModelGraph(P))\) by \cref{nodes-correct}, and also \(|siblings(I_1, I_2)| = 1\) by \cref{edge-neighbors}.
Therefore \((I_1, I_2) \in E(ModelGraph(P))\), so \(\st{(I_1, I_2)}{(H_1, I_1, H_2, I_2) \in E_{|\vec{H}|}} \subset E(ModelGraph(P))\).

Because \(\vec{H}_{1:|\vec{H}|} = pars(P)\), 
\begin{align*}
&\st*{ \big(I_1 \cap impls(\vec{H}_{1:|\vec{H}|}), I_2 \cap impls(\vec{H}_{1:|\vec{H}|})\big)}{\begin{aligned} &I_1, I_2 \in N(ModelGraph(P)),\\ &siblings(I_1, I_2) = \{(i_1, i_2)\},\\ &par(i_1) \in \vec{H}_{1:|\vec{H}|} \end{aligned}}\\
&\qquad = \st*{(I_1, I_2)}{I_1, I_2 \in N(ModelGraph(P)), |siblings(I_1, I_2)| = 1}\\
&\qquad = E(ModelGraph(P))
\end{align*}
So by \cref{edges-complete}, \(E(ModelGraph(P)) \subset \st*{(I_1, I_2)}{(H_1, I_1, H_2, I_2) \in E_{|\vec{H}|}}\)
\end{proof}

\subsubsection{Efficiency}
\label{sec:org047ec90}
For a modular program \(P\), let \(N = |N(ModelGraph(P))|\) and \(E = |E(ModelGraph(P))|\). Index the graph nodes as \(\st*{n_{i}}{i \in [1\dots N]} = N(ModelGraph(P))\) and the edges as \(\st*{(e_{i, 1}, e_{i, 2})}{i \in [1\dots E]} = E(ModelGraph(P))\).
Let \(H = |holes(P)|\).

The minimum size possible of the explicit representation of \(ModelGraph(P)\) is then \(\sum_{i \in [1\dots N]} |n_i| + \sum_{j \in [1\dots E]} |e_{j, 1}| + |e_{j, 2}|\) identifiers.

\begin{enumerate}
\item \textbf{Runtime efficiency}.
\label{sec:orgfda031b}
Consider the selection set insertion operations \(I\cup\{i\}\) in \texttt{expand-node}.
Since implementations are never removed from any selection sets, and selection sets are never removed from model prefix sets, the number of such insertions must be \(O(\sum_{i \in [1\dots N]} |n_i|)\).

Consider the hole union operations \(H \cup holes(i)\) in \texttt{expand-node}.
Since each \(h \in holes(i)\) is eventually replaced by exactly one implementation, the amortized cost is again \(O(\sum_{i \in [1\dots N]} |n_i|)\).

The same arguments apply to the edge operations in \texttt{new-edges} and \texttt{expand-edge}:
the number of implementation set insertions must be \(O(\sum_{j \in [1\dots E]} |e_{j, 1}| + |e_{j, 2}|)\), and the amortized cost of each hole union operation is also \(O(\sum_{j \in [1\dots E]} |e_{j, 1}| + |e_{j, 2}|)\).

The set inclusion checks \(h \in H\) result in a small constant overhead for \texttt{expand-node}, \texttt{new-edges} and \texttt{expand-edge}.
Since \texttt{expand} is called \(H\) times with at most \(N\) nodes and \(E\) edges, the number of inclusion checks is \(O(H * N + H * E)\).

The resulting runtime complexity is then \(O(H*N + H*N)\).

\item \textbf{Space efficiency.}
\label{sec:org481b46c}
Because the nodes \(N_{j}\) and edges \(E_{j}\) are the entire program state, \(|N_{j}| \leq |N(ModelGraph(P))|\) and \(|E_{j}| \leq |E(ModelGraph(P))|\), and the elements of \(N_j\) and \(E_j\) are less than or equal to the elements of \(N(ModelGraph(P))\) and \(E(ModelGraph(P))\), the program state never exceeds the size of the output, so space complexity is \(\Theta(\sum_{i \in [1\dots N]} |n_i| + \sum_{j \in [1\dots E]} |e_{j, 1}| + |e_{j, 2}|)\).

\item \textbf{Notes on efficiency}.
\label{sec:orgc94b31c}
Since \(N_{\vec{H}}\) does not depend on any \(E_{j}\), when we only calculate the graph nodes we have a runtime complexity of \(O(H*N)\) and space complexity of \(O(\sum_{i \in [1\dots N]} |n_i|)\).

This method can efficiently handle the extreme cases in \cref{fig:pathological}.
We have no no asymptotic time or space overhead for the example \cref{fig:patho:tall}, and for \cref{fig:patho:wide}, we have no space overhead but an asymptotic time overhead of \(O(H*H)\) set inclusion checks.
\end{enumerate}
\subsection{\texttt{ModelNeighbors}}
\label{sec:orgaecc377}
\label{neighbors-algorithm}

Let \(Nei(I,P)\) be a shorthand for our previously defined \texttt{ModelNeighbors}:, \(Nei(I,P) = \texttt{ModelNeighbors}(P, I) = \st*{I'}{(I, I') \in E(\texttt{ModelGraph}(P))}\).

\noindent
Let \(Limit(I, P)\) be equal to a modular program \(P\), except that, for all holes \(h \in pars(P)\):
\[impls_{Limit(I, P)}(h) =
   \begin{dcases} 
   impls_P(h) \cap I & \text{if } h \in pars(I) \\
   impls_P(h)        & \text{otherwise}
   \end{dcases}\]

\noindent
We can compute \(Nei\) efficiently with \(Limit\) and our previously defined \(ModelGraph\) algorithm:

\[Nei(I, P) = \bigcup_{i \in I} \bigcup_{i' \in impls(par(i)) - \{i\}} N(ModelGraph(Limit(I-\{i\}\cup \{i'\}, P)))\]

\subsubsection{Proofs of correctness and completeness}
\label{sec:orgdd80503}
\begin{lemma}[Limit(I, P) is a valid modular program]
\label{limit-valid}
\(valid_{structure}(P)\), \(valid_P(I)\) implies \(valid_{structure}(Limit(P, I))\).
\end{lemma}
\begin{proof}

\noindent
Requirements \cref{acyclic,unique-ids} are implied by \(valid_{structure}(P)\) and \(impls(Limit(P, I)) \subset impls(P)\).

Requirement \cref{every-hole} is \(\forall h \in holes(P) \cup holes(P_{base})\), \(impls_{Limit(P, I)}(h) \neq \emptyset\).
For every \(h\), if \(h \not\in pars(I)\), then \(impls_{Limit(P, I)}(h) = impls_P(h)\), which must be non-empty given \(valid_{structure}(P)\).
If \(h \in pars(I)\), then \(impls_{Limit(P, I)}(h) = impls_P(h) \cap I\), which must have one element by the definition of \(valid_P(I)\).
\end{proof}

\begin{lemma}[Neighbors of $I$ are valid under $Limit(I-\{i\}\cup \{i'\}, P)$.]
\label{neighbors-valid}
\(\forall I' \in N(ModelGraph(P))\) such that \(siblings(I, I') = \{(i, i')\}\), \(valid_{Limit(I-\{i\}\cup \{i'\}, P)}(I')\).
\end{lemma}
\begin{proof}

\noindent
\begin{enumerate}
\item \textbf{\(I' \subset impls(Limit(I-\{i\}\cup \{i'\}, P))\)}:
Suppose \(\exists i_-' \in I'\) such that \(i_-' \not\in impls(Limit(I-\{i\}\cup \{i'\}, P))\).
It must be that \(par(i_-') \in pars(I-\{i\}\cup\{i'\})\), because otherwise \(impls_{Limit(I-\{i\}\cup \{i'\}, P)}(i_-') = impls_P(i_-')\) by the definition of \(Limit\).
Then \(\exists i_- \in I-\{i\}\cup\{i'\}\ s.t.\ par(i_-) = par(i_-')\), so \((i_-, i_-') \in siblings(I-\{i\}\cup\{i'\}, I')\) and \((i_-, i_-') \neq (i, i')\).

Since \(i' \in I'\), \(i_- \neq i'\), so \(i_- \in I'\).
Then \((i_-', i_-) \in siblings(I, I')\), which contradicts \(siblings(I, I') = \{(i, i')\}\), so there is no such \(i_-'\).
\item \textbf{\(pars(I') = holes(I') \cup holes(Limit(I-\{i\} \cup \{i'\}, P)_{base})\)} follows from \(valid_P(I')\) and \cref{limit-valid} because \(Limit\) does not modify \(base\), \(holes\) or \(pars\).
\item \textbf{\(siblings(I', I') = \emptyset\)} by \(valid_P(I')\).
\end{enumerate}
\end{proof}

\begin{lemma}[Selections that are valid under $Limit(I-\{i\}\cup \{i'\}, P)$ are neighbors of I.]
\label{valid-neighbors}
\(\forall I' \in N(ModelGraph(Limit(I-\{i\}\cup \{i'\}, P))), valid_P(I')\) and \(siblings(I, I') = \{(i, i')\}\).
\end{lemma}
\begin{proof}
Because \(impls(Limit(I-\{i\}\cup \{i'\}, P)) \subset impls(P)\) and \cref{limit-valid}, \(valid_{Limit(I-\{i\}\cup \{i'\}, P)}(I')\) implies each requirement of \(valid_P(I')\).

Suppose \(siblings(I, I') = \emptyset\).
Because \(I, I' \in N(ModelGraph(P))\), by \cref{no-siblings-equal}, \(I = I'\).
Then, \(i\in I \implies i \in I'\).
However, by \(valid_{Limit(I-\{i\}\cup \{i'\}, P)}(I')\), \(I' \subset impls(Limit(I-\{i\}\cup \{i'\}, P)) \not\ni i\), which is a contradiction.
Therefore \(siblings(I, I') \neq \emptyset\).

Let \((i_2, i_2') \in siblings(I, I')\), so \(par(i_2) = par(i_2')\).
If \(par(i_2') \neq par(i)\), then by the definition of \(Limit\), \(impls_{Limit(I-\{i\}\cup \{i'\}, P)}(par(i_2')) = \{i_2\}\), so \(i_2 = i_2'\), which contradicts \((i_2, i_2') \in siblings(I, I')\).
If \(par(i_2') = par(i)\), then \(impls_{Limit(I-\{i\}\cup \{i'\}, P)}(par(i_2')) = \{i'\}\), so \((i_2, i_2')\) must be \((i, i')\).
Therefore, \(siblings(I, I') = \{(i, i')\}\).
\end{proof}

\begin{theorem}[The models of $Limit(I-\{i\}\cup \{i'\}, P)$ are exactly the neighbors of $I$ with $siblings(I, I') = \{i, i'\}$, so we can union over possible siblings]
\label{nei-correct}
\[Nei(I, P) = \bigcup_{i \in I} \bigcup_{i' \in impls(par(i)) - \{i\}} N(ModelGraph(Limit(I-\{i\}\cup \{i'\}, P))) \].
\end{theorem}
\begin{proof}
By \cref{valid-neighbors,neighbors-valid}, \(\st{I}{I \in N(ModelGraph(P)), siblings(I, I') = \{(i, i')\}} = N(ModelGraph(Limit(I-\{i\}\cup \{i'\}, P)))\).

\begin{align*}
\bigcup_{i \in I} & \bigcup_{i' \in impls(par(i)) - \{i\}} N(ModelGraph(Limit(I-\{i\}\cup \{i'\}, P)))\\
=& \bigcup_{i \in I} \bigcup_{i' \in impls(par(i)) - \{i\}} \st{I}{I \in N(ModelGraph(P)), siblings(I, I') = \{(i, i')\}} \\
=& \st{I}{I \in N(ModelGraph(P)), |siblings(I, I')| = 1} \\
=& Nei(I, P) \\
\end{align*}
\end{proof}
\subsubsection{Efficiency}
\label{sec:orgb5f8605}

Since we are using the \(ModelGraph\) implementation from \cref{modelgraph-algorithm}, the runtime complexity of N(ModelGraph(P)) is \(O(H * N)\) for \(H = holes(P)\) and \(N = |N(ModelGraph(P))|\).

Our runtime complexity is then:

\begin{align*}
O(Nei(I, P)) =& O(\bigcup_{i \in I} \bigcup_{i' \in impls(par(i)) - \{i\}} N(ModelGraph(Limit(I-\{i\}\cup \{i'\}, P))))\\
=& O(\sum_{i \in I} \sum_{i' \in impls(par(i)) - \{i\}} N(ModelGraph(Limit(I-\{i\}\cup \{i'\}, P))))\\
=& O(\sum_{i \in I} \sum_{i' \in impls(par(i)) - \{i\}} H |N(ModelGraph(Limit(I-\{i\}\cup \{i'\}, P)))|)\\
=& O(H \sum_{i \in I} \sum_{i' \in impls(par(i)) - \{i\}} |N(ModelGraph(Limit(I-\{i\}\cup \{i'\}, P)))|)\\
=& O(H |Nei(I, P)|)
\end{align*}

Our space complexity is the size of the output.
\section{Additional features}
\label{sec:org11d2da9}
\label{extensions}
This section presents two extensions of the module system described in \cref{mstan-syntax,semantics}, giving users more expressive power to build their desired networks of models.
\subsection{Append blocks}
\label{sec:org875fb61}
In the base syntax, we allow modules to define a \texttt{parameters} block that upon concretization is appended onto the base parameters block.
In the same way, we can allow modules to add on to each of the other Stan blocks (except for \texttt{data}, which is fixed for the whole network of models).
The extended syntax generalizes \texttt{MODULE\_IMPLEMENTATION\_M}:

\begin{verbatim}
MODULE_IMPLEMENTATION_M:
  module "impl_identifier" hole_identifier((TYPE identifier,)*) {
    FUNCTIONS_M?
    TRANSFORMED_DATA_M?
    TRANSFORMED_PARAMETERS_M?
    PARAMETERS?
    MODEL_M?
    GENERATED_QUANTITIES_M?
    STMT_LPDF_M;*
    return EXPR_M;?
  }
\end{verbatim}

We refer to these module blocks as append blocks because they are appended to the end of their corresponding Stan blocks upon concretization.

The rules for scope and available effects are the same for within-module blocks as they are for the corresponding Stan blocks as described in \cref{stan-effects-scope}.
For example, variables defined in a module \texttt{transformed parameters} block can be referenced inside of that module's \texttt{model} block but not its \texttt{transformed data} block.
Module arguments are not available inside append blocks.

ApplyImpl is updated to include \texttt{InlineFunction} calls for each append block in the same way that it is used for \texttt{parameters}.

Concretization is still correct:
\begin{itemize}
\item Scope is preserved: For any statement in a module that could reference a global variable \texttt{g} in module append block \texttt{B}, \texttt{g} will still be a valid reference after concretization because \texttt{B} must also be included.
For any statement \texttt{S} in a module append block that references a global variable \texttt{h} in its corresponding base block, \texttt{h} will still be a valid reference after concretization because \texttt{M} will be inserted after the declarations of the base block.
\item Type correctness and effect correctness are preserved because appends are statements and must have the same effect requirements as the blocks into which they are inserted, as usual.
\item All other arguments about statement order, well-typedness, and lack of cycles are unaffected by this change.
\end{itemize}

Since the inclusion of append blocks has no affect on the module dependency graph, this change does not change the correctness of the network of models or neighbor algorithms.

Use cases of append blocks include:
\begin{itemize}
\item Defining functions for use in the module and as arguments to its descendants.
\item Defining transformed data or transformed parameters that are used in the module.
\item Adding a prior distribution to the model block for each parameter defined the module, without requiring that the module signature have the \texttt{LPDF} effect.
\item Emitting the likelihood of the data under a distribution defined the module as a generated quantity (as in the case study in \cref{birthday}).
\end{itemize}

\subsection{Module fields}
\label{sec:org17442f5}
Sometimes, two or more behaviors only make sense to include together.

For example, suppose we want to define a change of variables transformation.
We need to define a \texttt{transform} function and corresponding \texttt{inverse} function.
Suppose we want to abstract the transformation into a hole so that it can be swapped out.
If we were to define separate holes for \texttt{Transform} and \texttt{Inverse}, we may end up with implementations may not match, because they can be selected independently.
We would rather have one hole, \texttt{Transformation}, that packages compatible \texttt{transform} and \texttt{inverse} implementations together.

We can achieve this coupling by letting modules contain multiple named behaviors call \emph{fields}.
With fields, we would defined one hole, \texttt{Transformation}, with two fields, \texttt{forward} and \texttt{reverse}, referenced in code as \texttt{Transformation.forward} and \texttt{Transformation.reverse}.

Fields of a module also share the append blocks (and therefore global scope) of a module.

The syntax for associated behavior introduces a second, optional variant of \texttt{MODULE\_IMPLEMENATION\_M} where arguments are moved into field declarations:

\begin{verbatim}
MODULE_IMPLEMENTATION_M: module "impl_identifier" hole_identifier {
  FUNCTIONS_M?
  TRANSFORMED_DATA_M?
  TRANSFORMED_PARAMETERS_M?
  PARAMETERS?
  MODEL_M?
  GENERATED_QUANTITIES_M?

  field field_identifier?(TYPE identifier*) {
    STMT_LPDF_M;*
    return EXPR_M;*
  } +
}
\end{verbatim}
We use \texttt{impl\_identifier} to stand in for valid field names, while \texttt{field} is a new keyword.

Each of a hole's fields must be implemented by a corresponding field block and implementations as though it were a new hole.
A field \(\texttt{f}_{\texttt{i}}\) in hole \(\texttt{h}\) is referenced as \(\texttt{h.f}_{\texttt{i}}\texttt{(..)}\), except when the field identifier \(\texttt{f}_{\texttt{i}}\) is empty, in which case the reference is \(\texttt{h(..)}\). 

Each of a hole's fields is treated like a separate hole for the purposes of concretization, and they have no differences in terms of the correctness of the resulting Stan program.
Holes with fields are treated like holes without fields for the purpose of building the module graph, model graph and neighbor sets; an implementation has a dependency on a hole \texttt{h} if any of its fields depends on \texttt{h}.

Coupling behavior into the fields of a module is only one constraint that we could impose on co-selection of implementations.
We discuss a more general constraint logic in \cref{additional-features}.
\section{Macros}
\label{sec:org2e671ad}
\label{macros}
While the module system described so far is in theory flexible enough to describe arbitrarily complex networks of models, it may sometimes be too verbose to be practical.

For instance, suppose we want to build a regression model, but we don't know which of \(N\) features to include. 
Our model space of interest then consists of one model for each subsets of features.
Since we want to choose inclusion or exclusion for each feature, we will need at least \(N\) holes, each with two implementations, so we will need to write \(2N\) implementations.
This is cumbersome and repetitive for large \(N\).

To make the language more expressive, we define a series of ``macros'' as shorthand ways of generating larger modular programs.
Each macro is translated by the compiler back into the basic module language, so we can then use the same algorithms defined in \cref{algorithms}.

\Cref{macro-table} is a summary of the macros described in the following sections.

\begin{table}[htbp]
\centering
\begin{tabular}{lll}
Name & Syntax & Description\\
\hline
Collection hole & \texttt{H+} & Select a subset of implementations rather than one\\
Indexed hole & \texttt{H[i..j]} & Copy implementations of \(H\) for each index in range\\
Hole instance & \texttt{H<j>} & Copy a hole; same selection but new parameters\\
Hole copy & \texttt{H<<j>>} & Copy a hole; independent selection\\
Ranged versions & \texttt{H<i..j>} & Apply the macro for each index in range,\\
 & \texttt{H<<i..j>>} & \quad collecting an array of results\\
Multi-ranges & \texttt{i..j,k..l,..} & Same as a range, but for each index combination\\
Range exponent & (\texttt{i..j)\textasciicircum{}e} & Same as a range, but \texttt{e} indices without replacement\\
Hole product & \(\texttt{H}_1\texttt{*}\texttt{H}_2\texttt{(..)(..)}\) & New hole with implementations \(impls(\texttt{H}_1) \times impls(\texttt{H}_2)\)\\
Hole exponent & \texttt{H\textasciicircum{}j} & Same as a hole product but without replacement\\
\end{tabular}
\caption{\label{macro-table}Summary of macros}

\end{table}

The expansion of each macro transforms the program's module graph in some way, it may synthesize new modules with new (Stan) code, and it may define a translation of selection strings (\cref{selection-string}) to apply to the post-expansion module graph.
\subsection{Collection holes}
\label{sec:org030720f}
\label{collection-holes}
A ``collection hole'' is a hole that can be filled with any number of implementations, rather than exactly one.

Consider our regression model example where we want to use a subset of \(N\) features.
For each feature, we need a hole with one implementation that includes the feature and another that does not.
Then, we need to collect all of the holes into an array to pass into the regression.
The ``collection hole'' macro automates this pattern, so that the user can instead write one ``collection hole'' with \(N\) implementations instead of \(N\) holes and \(2N\) implementations.

Collection holes are identified by a \texttt{+} at the end of a hole identifier.
The value of a non-\texttt{void} collection hole is an array containing the values of the selected implementations in an undefined order.

\Cref{collection-trans} shows how expansion of collection holes modifies a program's module graph.

\begin{figure}[htbp]
\centering
\includegraphics[width=360pt]{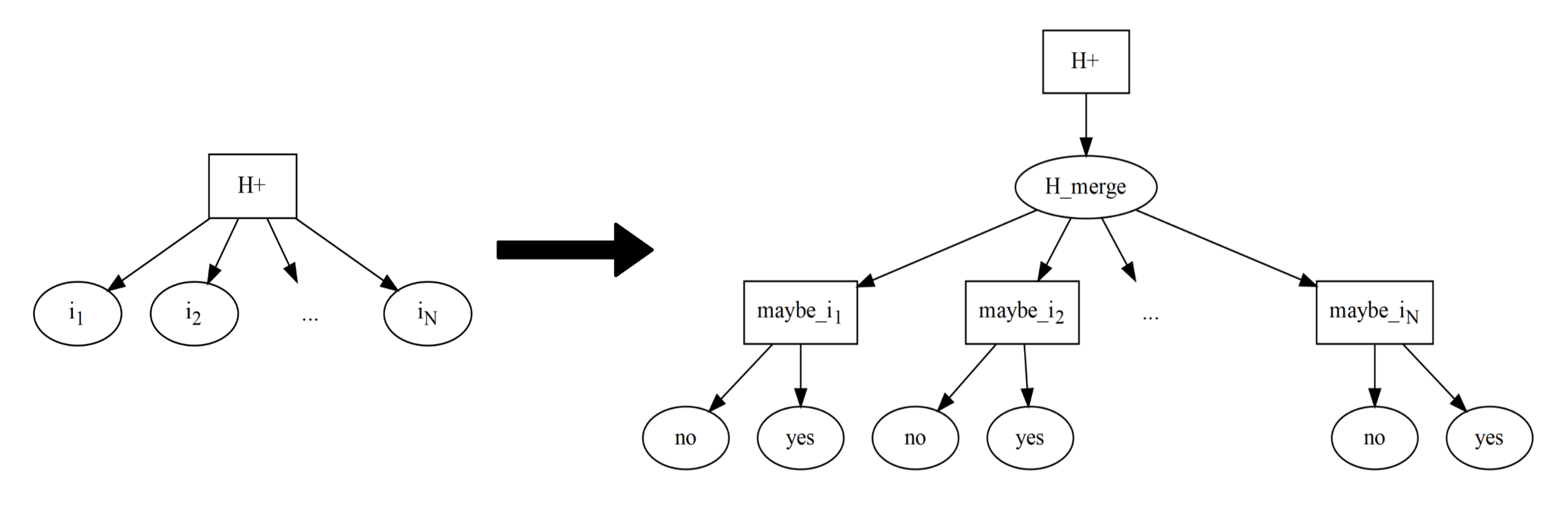}
\caption{\label{collection-trans}Collection hole module graph transformation. A collection hole \texttt{H+} produces a modified hole that lets a user select which implementations to include (\texttt{yes}) or exclude (\texttt{no}), and then concatenates the results (\texttt{H\_merge}).}
\end{figure}

For each implementation \texttt{i} of \texttt{h+}, the implementation \texttt{yes} is the same as \texttt{i} but wraps its result in a singleton array, while \texttt{no} does nothing and returns an empty list.
The implementation \texttt{merge\_h} returns a concatenation of the arrays returned by each new hole.

Users can then select a subset of implementations in their selection strings, as in \texttt{..h:[i\_1,i\_2,..]..}.
If a collection hole \texttt{h} has a set of implementations \texttt{C}, then a selection string \(\texttt{s}_{before},\texttt{h:[I],s}_{after}\) for some list of implementations \(I\subset C\) is translated to \(\texttt{s}_{before}\texttt{,}\) \(\texttt{h:merge\_h,}\) \(\bigcup_{i\in I}\texttt{h\_i:yes,}\) \(\bigcup_{i\in C-I}\texttt{h\_i:no,}\) \(\texttt{s}_{after}\).

The network produced by a collection hole includes an edge between two nodes if they differ by exactly one inclusion\footnote{Like a Hamming graph.}.
\subsection{Indexed holes}
\label{sec:org8b27dd6}

An ``indexed hole'' is a hole that can generate additional implementations.

Indexed holes are identified by adding a range \texttt{[i..j]}, where \texttt{i} and \texttt{j} non-negative integer literals, to at the end of a hole identifier.
Implementations of indexed holes accept an index as an extra argument, denoted in brackets at the end of the hole identifier in the their definition, such as in \texttt{module "i" h[j](..) \{..\}}.
Here, \texttt{j} is the index, and can be used as an integer literal within the module, because it will be replaced with by each integer in the range upon macro expansion.
In this way, each written module implementation serves as a template for generating more.

\Cref{indexed-trans} shows how expansion of indexed holes modifies a program's module graph.

\begin{figure}[htbp]
\centering
\includegraphics[width=360pt]{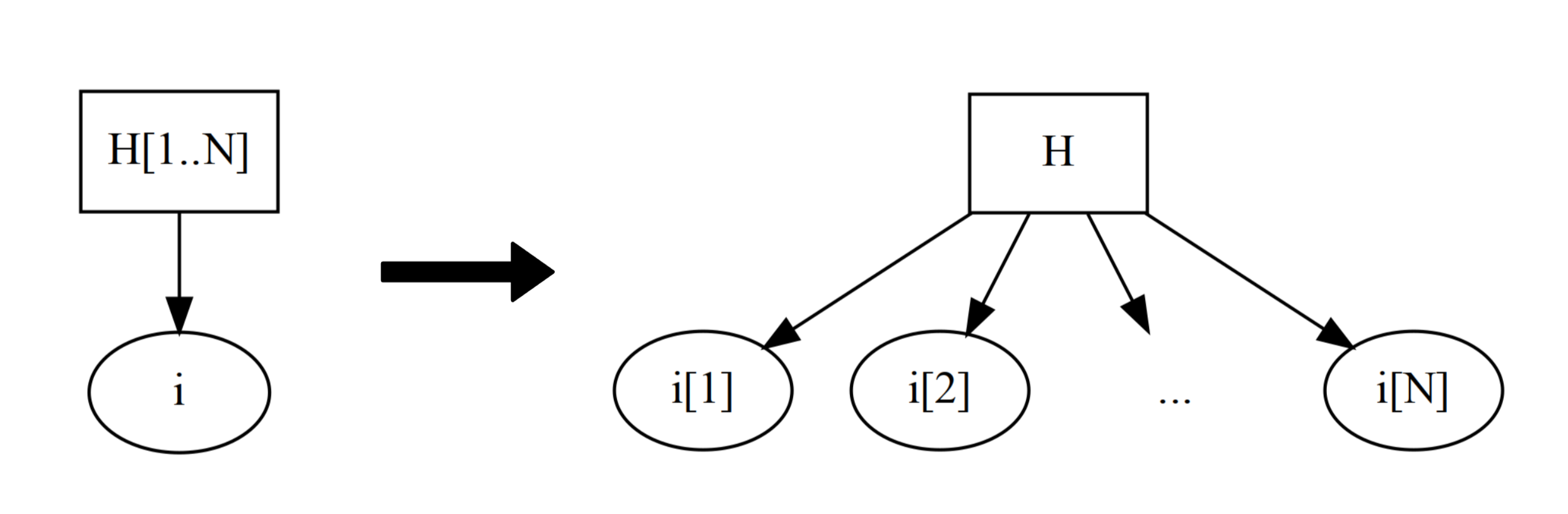}
\caption{\label{indexed-trans}Indexed hole module graph transformation. An indexed hole \texttt{H[1..N]} produces \texttt{N} copies of its implementation.}
\end{figure}

Users can then specify an indexed implementation to fill the indexed hole, as in \texttt{..h:i[5]..}.

\subsection{Hole instances and hole copies}
\label{sec:orgda3cbb4}
So far, the design of holes effectively assumes that they represent decisions about individual subcomponents of a model. For example, when a parameter is defined within a hole, it is only ever translates to single parameter in the resulting program.
It may be instead that a hole should be repeated in multiple places.

For example, suppose a hole \(h\) represents a model of a storm cloud.
What if we have data from two storm clouds?
There are three possibilities:
\begin{enumerate}
\item We only want one copy of \(h\), and data from both clouds will be used to estimate the parameters in \(h\).
\item We want two copies of \(h\), one to model each cloud, with identical implemenations but separate parameters.
\item We want two copies of \(h\), one to model each cloud, but which may have different implementations.
\end{enumerate}

Case 1 corresponds to the basic semantics: users can call \(h()\) multiple times, but all references to \(h\) are filled with the same implementation, and the append blocks of \(h\)'s implementation are only added once, so all references use the same parameters.

Case 2 is the motivation for a macro called \emph{hole instances}.

Case 3 is the motivation for a macro called \emph{hole copies}.

\subsubsection{Hole instances}
\label{sec:org6323613}
Hole instances are identified by the syntax \texttt{hole\_identifier<j>}, where \texttt{j} is a non-negative integer literal.
Hole instances are transformed in the module tree in the following way:

\Cref{instance-trans} shows how expansion of hole instances modifies a program's module graph.

\begin{figure}[htbp]
\centering
\includegraphics[width=360pt]{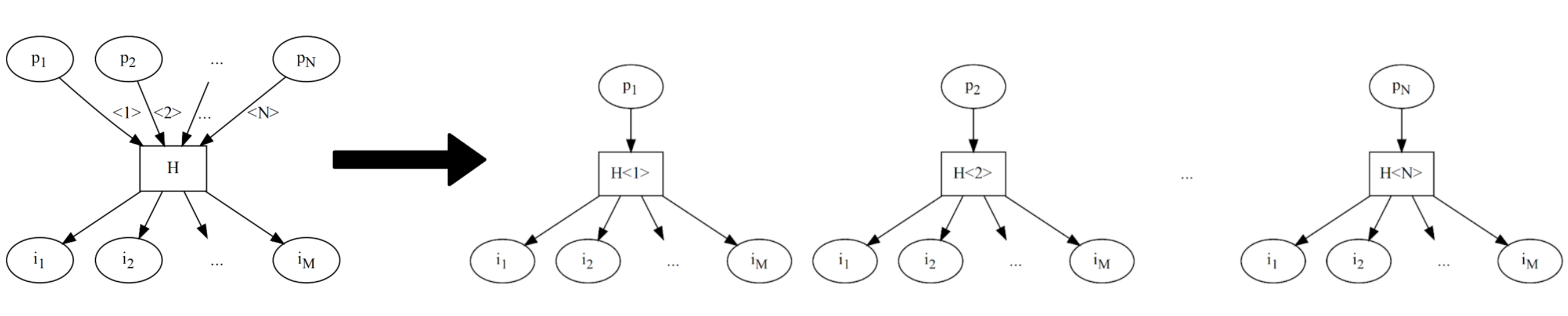}
\caption{\label{instance-trans}Hole instance module graph transformation. When some code \(\texttt{p}_i\) refers to \texttt{H<i>}, a new hole called \texttt{H<i>} is created with copies of the implementations of \texttt{H}.}
\end{figure}

Here, \(\texttt{h[1]}\dots \texttt{h[N]}\) are identical copies of \(h\) except that their local and global variables are given unique names.
This way, there are \(N\) copies of the implemenation of \(h\), each refering to its own set of new global variables such as parameters.
Each \texttt{j} found in a reference \texttt{hole\_identifier<j>} produces a new instance.

Users specify selections in the same way as if \(h\) were not copied: \texttt{..,h:i,..}. 
A selection string \(\texttt{s}_{before},\texttt{h:i},\texttt{s}_{after}\) is translated into \(\texttt{s}_{before},\bigcup_{j\in 1\dots N}\texttt{h}_{j\texttt{:}i},\texttt{s}_{after}\).

Optionally, module implementations can accept an index as an extra argument, denoted in angle brackets at the end of the hole identifier in their definition, such as in \texttt{module "i" h<j>(..) \{..\}}.
When that module implementation is used as a hole instance, the variable \texttt{j} can be used as an integer literal.
This way, \texttt{j} can specify hole instances within the module, copying deeper into the module graph.

\subsubsection{Hole copies}
\label{sec:org8ecc278}
Hole copies are identified by the syntax \texttt{hole\_identifier<<j>>}, where \texttt{j} is a non-negative integer literal.
Hole copies behave the same way as hole instances, except that implementations must be selected independently for each generated hole \texttt{h<<j>>}, as in:
\(\texttt{..,h<<1>>:i}_{\texttt{1}},\dots,\texttt{h<<N>>:i}_{\texttt{N}}\texttt{,..}\).

\subsubsection{Ranged hole instances and copies}
\label{sec:org588502d}
To generate many hole instances or copies automatically, \texttt{h[1:N](..)} is translated into a Stan array \texttt{\{ h[1](..), h[2](..), .., h[N](..) \}} and \texttt{h<1:N>(..)} is translated into \texttt{\{ h<1>(..), h<2>(..), .., h<N>(..) \}}.
\subsection{Multi-ranges and range exponents}
\label{sec:org8ce8dab}
\label{ranges}
Macros that can have ranges, namedly \texttt{H[i..j]}, \texttt{H<i..j>}, and \texttt{H<<i..j>>}, can also accept multi-ranges.
Multi-ranges are the same as ranges except that they produce one result per combination of their ranges.
For example, a \texttt{1..3,1..5} produces \texttt{(1,1),(1,2),..(3,5)}.
Implementations that accept indices as extra arguments, such as \texttt{h<j>}, must then accept \texttt{h<i,j>}.

Ranges exponents come in three variants.
\texttt{R\textasciicircum{}n} is equivalent to a multi-range with \texttt{R} repeated \texttt{n} times, for example, \texttt{(1..3)\textasciicircum{}2} is equivalent to \texttt{1..3,1..3}.
\texttt{R\textasciicircum{}Pn} is like \texttt{R\textasciicircum{}n} except that it gives ordered permutations without replacement, for example \texttt{(1..3)\textasciicircum{}P3} does not include \texttt{(1,1)}, \texttt{(2,2)}, or \texttt{(3,3)}.
\texttt{R\textasciicircum{}Cn} is like \texttt{R\textasciicircum{}Pn} except that it gives unordered combinations without replacement, for example \texttt{(1..3)\textasciicircum{}C3} does not include \texttt{(2,1)}, \texttt{(3,1)}, or \texttt{(3,2)}.

Multi-ranges and range exponents make it easier to generate holes and implementations that represent combinations.
\subsection{Hole products and hole exponents}
\label{sec:org0836312}
A hole product \(\texttt{H}_1\texttt{*H}_2\texttt{(.., ..)}\) is a hole that combines the implementations of the holes \(\texttt{H}_1\) and \(\texttt{H}_2\).
For each \(i_1 \in impls(\texttt{H}_1)\) and \(i_2 \in impls(\texttt{H}_2)\), \((i_1,i_2)\) is an implementation of \(\texttt{H}_1\times \texttt{H}_2\) that returns a tuple of the results of \(i_1\) and \(i_2\)
The arguments lists of \(H_1\) and \(H_2\) are concatenated together to make the argument list of \(\texttt{H}_1\texttt{*H}_2\).

\Cref{product-trans} shows how expansion of hole products modifies a program's module graph.

\begin{figure}[htbp]
\centering
\includegraphics[width=360pt]{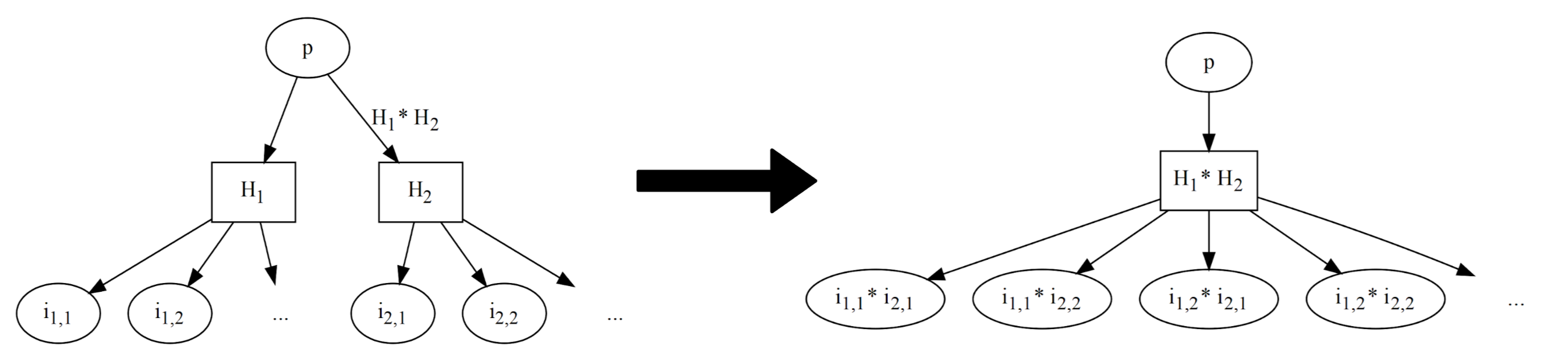}
\caption{\label{product-trans}Hole product expansion shown as a module graph transformation. When some code \texttt{p} refers to a product \(\texttt{H}_1\texttt{*H}_2\), a new hole is created with the cartesian product of \(\texttt{H}_1\) and \(\texttt{H}_2\)'s implementations.}
\end{figure}

Hole exponents are analogous to range exponents from \cref{ranges}.
\texttt{H\textasciicircum{}n} is equivalent to the product of \texttt{H} with itself \texttt{n} times.
\texttt{H\textasciicircum{}Pn} gives the ordered \texttt{n}-permuations of \(impls(\texttt{H})\) without replacement, while \texttt{H\textasciicircum{}Cn} gives the unordered \texttt{n}-combinations of \(impls(\texttt{H})\) without replacement.
\subsection{Example application of macros}
\label{sec:orgc99cabf}
Recall the regression example in which we want to include some subset of \(N\) variables in our model. Suppose \(N=100\).
To encode this model space in the module system requires \(2N=200\) module implementations.
Using an indexed collection hole, we can reduce this to one handwritten module implementation.

\begin{verbatim}
data {
    int N;
    matrix[100, N] x;
    vector[N] y;
}
parameters {
    real sigma;
}
model {
    y ~ normal(sum(Feature[1..100]+(x)), sigma);
}

module "f" Feature[n](x) {
    parameters {
        real theta;
    }
    return theta*x[n,:];
}
\end{verbatim}

This program represents a family of regression models on \texttt{y} given the features \texttt{x}, each including a different subset of features.
\texttt{y} is modeled with a normal distribution centered on the sum of the subset of features, \texttt{Feature[1..100]+(x)}, which is an indexed collection hole: the range \texttt{[1..100]} copies the \texttt{Feature} module implementation for \texttt{n=1} to \texttt{100}, and the \texttt{+} indicates that each implementation is either included or excluded from the result.
An individual Stan program can be generated from this family by supplying a selection string, such as: \texttt{Feature:[1,2,3]}, which includes only the first three features.

Now suppose we also want our regression to include some subset of 2- and 3-way interactions between variables, such as \texttt{x[3]*x[9]} or \texttt{x[4]*x[10]*x[99]}.
To encode this model space requires \(2*(100 + {100 \choose 2} + {100 \choose 3}) = 333500\) module implementations.
We can use macros to reduce this to two or three handwritten module implementations.
We will show two alternative implementations.

The first way makes use of range exponents:
\begin{verbatim}
data {
    int N;
    matrix[100, N] x;
    vector[N] y;
}
parameters {
    real sigma;
}
model {
    y ~ normal(sum( Feature[1..100]+(x) )
                 + sum( FeaturePair[(1..100)^C2]+(x) ),
                 + sum( FeatureTriplet[(1..100)^C3]+(x) ),
               sigma);
}

module "f" Feature[n](x) {
    parameters {
        real theta;
    }
    return theta * x[n];
}

module "fp" FeaturePair[n, m](x) {
    parameters {
        real theta;
    }
    return theta * x[n] .* x[m];
}

module "ft" FeatureTriplet[n, m, p](x) {
    parameters {
        real theta;
    }
    return theta * x[n] .* x[m] .* x[p];
}
\end{verbatim}
This program uses indexed collection holes with range exponentials to collect and sum a subset of the features, feature pairs, and feature triplets.
An individual Stan program can be generated by supplying a selection string such as: \texttt{Feature:[1,2,3],} \texttt{FeaturePair:[(1,2),(1,4)],} \texttt{FeatureTriplet:[(1,2,3),(4,10,99)]}.

The second way makes use of hole products and exponents:
\begin{verbatim}
data {
    int N;
    matrix[100, N] x;
    vector[N] y;
}
parameters {
    real sigma;
}
model {
    vector[N] total = rep_vector(0, 100);
    for ((t, r) in Theta*Col[1..100]+()) {
      total += t * r;
    }
    for ((t, r1, r2) in Theta*Col[1..100]^C2+()) {
      total += t * r1 .* r2;
    }
    for ((t, r1, r2, r3) in Theta*Col[1..100]^C3+()) {
      total += t * r1 .* r2 .* r3;
    }

    y ~ normal(total, sigma);
}

module "t" Theta() {
  parameters {
    real theta;
  }
  return theta;
}

module "r" Col[n]() {
  return x[n];
}
\end{verbatim}
By taking a product \texttt{Theta*Col[1..100]}, we are generating an implementation of \texttt{Col} for each index 1 to 100, and then producing a new parameter for each of those implementations.
By taking a product of \texttt{Theta} with the exponent \texttt{Col[1..100]\textasciicircum{}2}, we are producing a new parameter for each pair of indices.
Each of these products returns an array of tuples of values which must then be multiplied and summed into the \texttt{total} vector.
The selection strings for this program look like:
\texttt{Theta*Col:[(t,1),(t,2),(t,3)],} \texttt{Theta*Col\textasciicircum{}2:[(t,1,2),(t,1,4)],} \texttt{Theta*Col\textasciicircum{}3:[(t,1,2,3),(t,4,10,99)]}.

These programs have \(\approx 333500\) module implementations and represent networks with \(2^{100 + {100 \choose 2} + {100 \choose 3}} \approx 10^{50000}\) models.
How can we use such large programs and spaces?
By never explicitly representing them.
Macros instantiate synthetic modules lazily, only when they are selected.
Large networks can be explored effiently by only enumerating neighbors with \texttt{ModelNeighbors}; in this case, the network diameter and branching factor are a more managable 166750.
\section{Example case studies}
\label{sec:org83ade2b}
In this section we present two small but real-world probabilistic modeling case studies that we have translated into the modular Stan language, and we discuss their benefits. 
These case studies showcase only two of the motivating use cases listed in \cref{intro}; the rest we for future work.

In addition, we present a web interface that can be used to follow along with our two examples.
\subsection{Interactive web interface}
\label{sec:orgbfa72b0}

We have built a prototype web interface for development and interactive visualizations of modular Stan \cite{website}.
\Cref{fig:interface} shows its interface.

\begin{figure}[htbp]
\centering
\includegraphics[width=.9\linewidth]{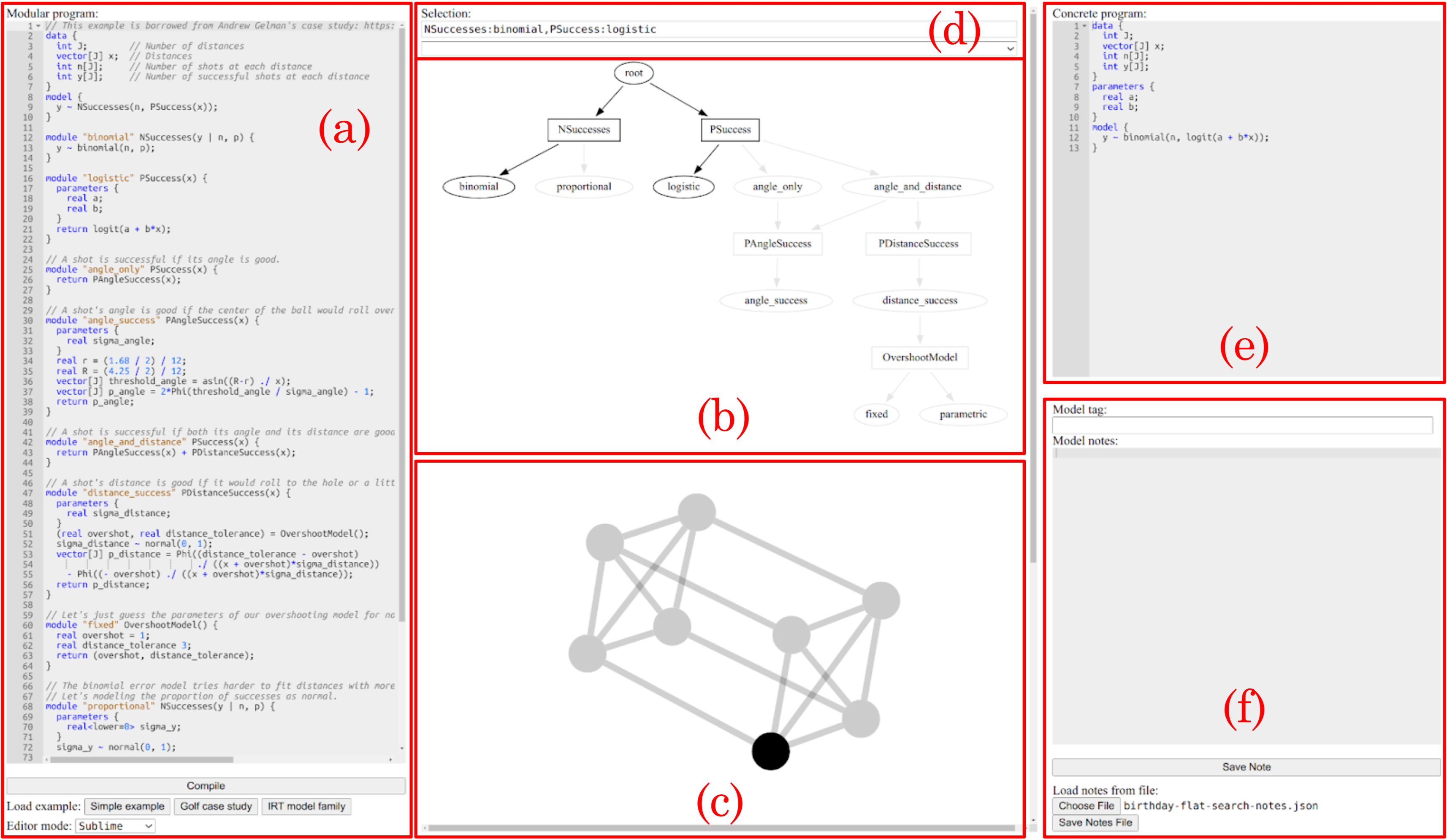}
\caption{\label{fig:interface}A labeled screenshot of the prototype web interface for modular Stan.}
\end{figure}

Users can write a modular Stan program or load an example program and compile it at (a).
When they do, interactive visualizations are produced: (b) the module graph and (c) the model graph.

The page keeps track of a module selection set. Users can modify the selection by: selecting or deselecting implementations the module graph (b), selecting complete models in the model graph (c), editing the selection string directly (d), or selecting a previously labeled model (d).
When the selection is modified, the model graph (c) highlights nodes compatible with that selection, and when the selection is valid, the corresponding concrete Stan program is displayed (e) and labels and notes associated with that program can be edited (f). The set of model labels and notes can be saved and loaded as a text file.

Users can also bookmark and annotate nodes in model graph.
Annotations can be saved and loaded as files separate from the modular Stan file, in a format that maps between the model's unique selection set and model labels and annotations.

\bigskip

Interactive versions of the following two case studies can be found at \cite{website-golf} and \cite{website-birthday}.

\subsection{``Golf'' case study: Modular Stan for ease and clarity of development}
\label{sec:orge5cd027}
    \label{golf}
This section gives a basic demonstration of how modular Stan can cleanly support and express a typical model development workflow, as an example of application \cref{tracking}.

The ``Golf'' case study \cite{golf-case-study} follows the development of a Bayesian statistical model for describing the probability that a professional golfer will sink a shot given their distance from the hole.

We represent the modeling process as a single modular Stan program.
The base of the program is the part that remains constant throughout development:

\begin{verbatim}
data {
  int J;        // Number of distances
  vector[J] x;  // Distances
  int n[J];     // Number of shots at each distance
  int y[J];     // Number of successful shots at each distance
}
model {
  y ~ NSuccesses(n, PSuccess(x));
}
\end{verbatim}
The \texttt{data} block describes \texttt{J} distances, where the \texttt{j}th distance is \texttt{x[j]} feet, and \texttt{y[j]} shots out of \texttt{n[j]} were successful.

The \texttt{model} block describes an abstracted modeling approach: we model the number of successes \texttt{y} as being drawn from some distribution \texttt{NSuccesses} parameterized by the number of attempts \texttt{n} and the probability of success \texttt{PSuccess}, which itself is a function of the distance \texttt{x}.
\texttt{NSuccesses} and \texttt{PSuccess} are holes.

A natural distribution to choose for \texttt{NSuccesses} is the binomial distribution. We express this as a module:

\begin{verbatim}
module "binomial" NSuccesses(y | n, p) {
  y ~ binomial(n, p);
}
\end{verbatim}

A simple way to take a real value like \texttt{x} to a probability is the \texttt{logit} function, so we choose to explore that option first for \texttt{PSuccess}:

\begin{verbatim}
module "logistic" PSuccess(x) {
  parameters {
    real a;
    real b;
  }
  return logit(a + b*x);
}
\end{verbatim}

If we stop here, we have only one choice of implementation for each of our holes, so our modular Stan program defines only one valid Stan program:

\begin{verbatim}
data {
  int J;        // Number of distances
  vector[J] x;  // Distances
  int n[J];     // Number of shots at each distance
  int y[J];     // Number of successful shots at each distance
}
parameters {
  real a;
  real b;
}
model {
  y ~ binomial(n, logit(a + b*x));
}
\end{verbatim}
This program implement logistic regression, which is the first model explored in the case study.

The next step in the probabilistic workflow or Box's loop is to criticize our model by applying it to the data.
We find that the model fit is lacking and the parameters \texttt{a} and \texttt{b} have no obvious physical interpretation.

This criticism motivates us to try a more sophisticated, more physically realistic implementation for \texttt{PSuccess}.
If we suppose that a shot will be successful if its trajectory angle is sufficiently precise, and we suppose that the angle is normally distributed, then we can write our model in terms of angle variance:

\begin{verbatim}
// A shot's angle is good if the center of the ball would roll
// over the hole.
module "angle_success" PSuccess(x) {
  parameters {
    real sigma_angle;
  }
  real r = (1.68 / 2) / 12; // ball radius
  real R = (4.25 / 2) / 12; // hole radius
  vector[J] threshold_angle = asin((R-r) ./ x);
  vector[J] p_angle = 2*Phi(threshold_angle / sigma_angle) - 1;
  return p_angle;
}
\end{verbatim}

We find this model has superior fit and interpretability and we continue iterating by adding on modules in this fashion.

The completed representation of the case study can be found with our source code \cite{golf-module-source} or at the web interface along with visualizations \cite{website-golf}.
Here it is presented by the web interface:

\begin{figure}\centering
\begin{subfigure}{0.495\textwidth}
\begin{center}
\includegraphics[width=.9\linewidth]{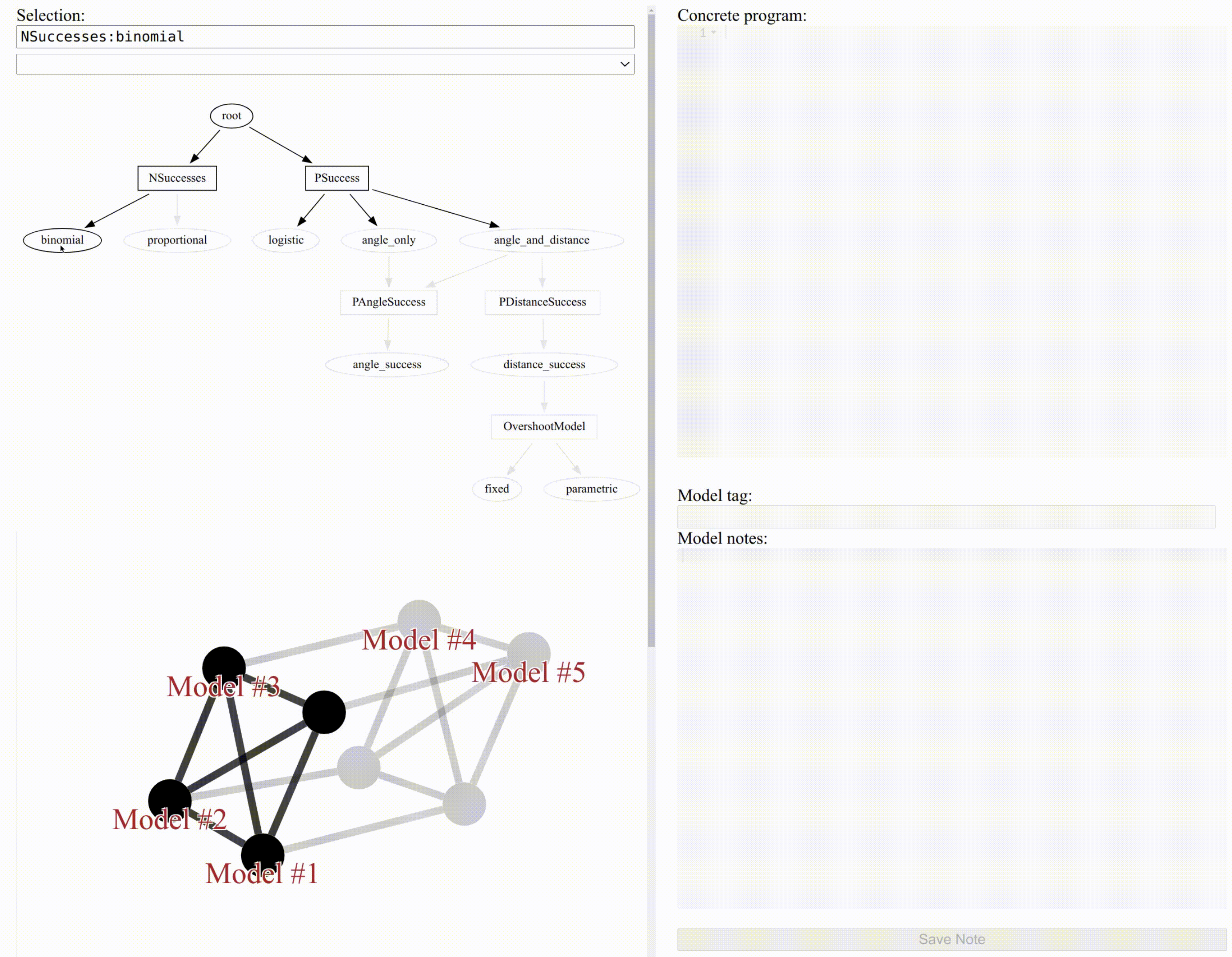}
\end{center}
\caption{}
\label{fig:golf:a}
\end{subfigure}
\begin{subfigure}{0.495\textwidth}
\begin{center}
\includegraphics[width=.9\linewidth]{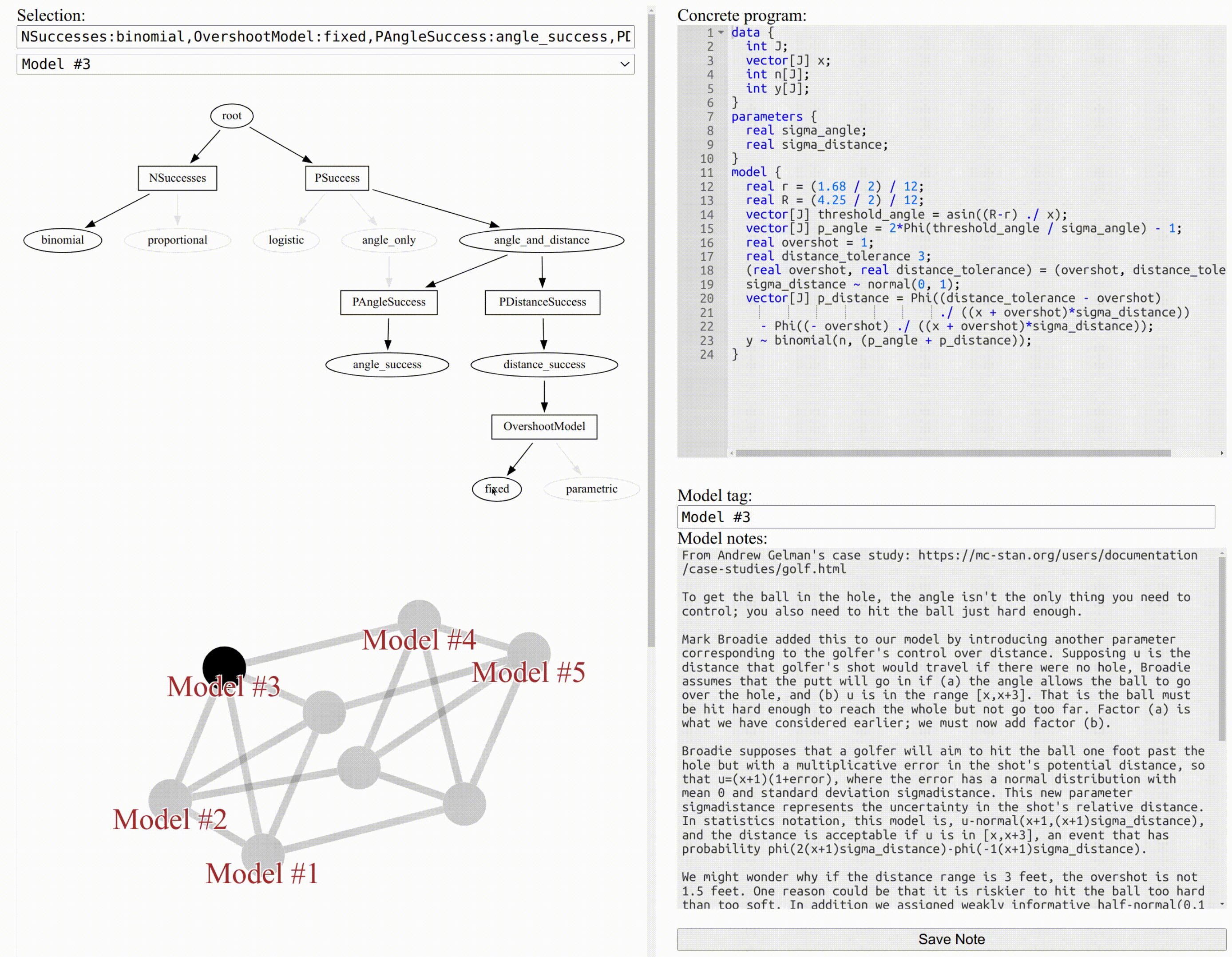}
\end{center}
\caption{}
\label{fig:golf:b}
\end{subfigure}
\end{figure}

Though this is a small example, we already see some benefits to clarity:
\begin{itemize}
\item There is only a single, minimal source file;
\item We have a standard, integrated way to draw the development path and document the decision-making evidence and rationale at each step;
\item The modular organization makes the solution space easier to understand and extend.
\end{itemize}

\subsection{``Birthday'' case study: Modular Stan as a platform for automation}
\label{sec:orgb44b2fa}
\label{birthday}

This section gives an example of how the network of models provides a platform for automation (use case \cref{automation}).
We demonstrate constructing a network, defining an evaluation metric, and performing a simple graph search.
We use a moderately sized network of models from a case study of incremental model improvements.
This approach would also apply to other multiple-model contexts, such as feature selection and symbolic regression.

The ``Birthday'' case study \cite{birthday-case-study,bda} follows the development of a statistical model of the number of babies born in the US on a given day, given birth data from 1969-1988.
The authors used a time-series Gaussian process approach.

Like many probabilistic modeling case studies and publications, the Birthday case study presents a series of models that differ by incremental variations.
The authors explicitly evaluate nine models, but the model variations they present implicitly define a much larger set of models that it would be reasonable to explore: what if a different combination of variations were applied, or in a different way?
To feasibly explore that larger set of models, we need automation, and for automation, we need an explicit representation of the model space, like the network of models offered by a modular Stan program.

We start by translating the case study into a single modular Stan program.
The translation reduces the number of lines of code from 1098 to 270 while increasing the number of models represented from nine to 120.

The translation process is largely mechanical, and involves encapsulating the variation between the given models into modules.
For example, model 2 adds a days-of-the-week Gaussian process component onto on model 1, model 3 adds a long-term-trend component onto model 2, etc.
Each of these variations becomes one or more module in the modular Stan program.
The full translation of the modular program can be found with our source code \cite{birthday-module-source} or at the web interface along with visualizations \cite{website-birthday} \footnote{Our modular Stan program would have been significantly more concise if we used the collection hole feature described in \cref{collection-holes}, but they were not implemented in our prototype compiler at the time of writing.}.

\begin{figure}[htbp]
\centering
\includegraphics[width=4.5in]{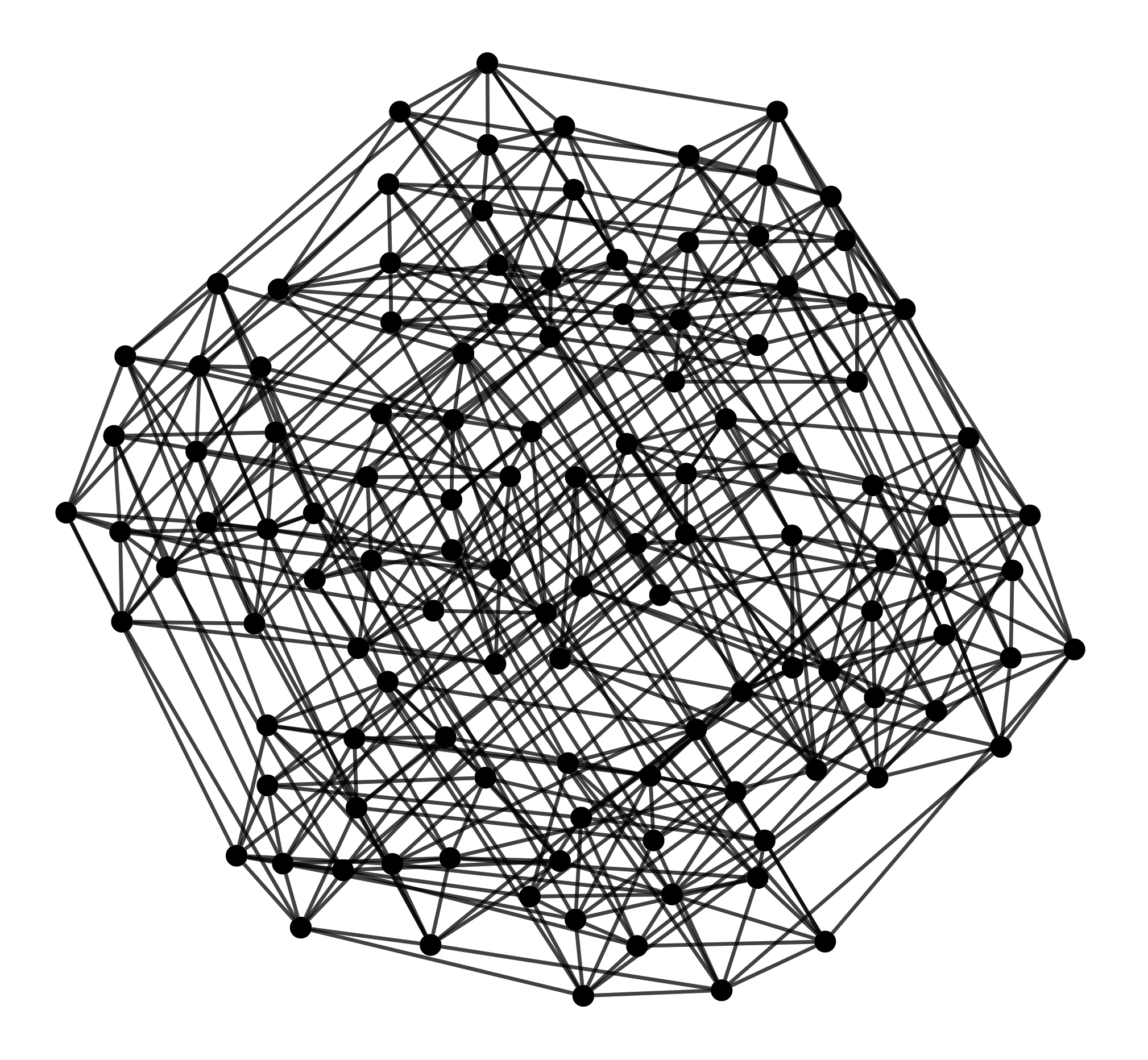}
\caption{\label{fig:birthdaygraphs}The model graph corresponding to the Birthday problem modular Stan program. Each node represents a model and each edge represents one swapped-out module.}
\end{figure}

\Cref{fig:birthdaygraphs} shows the model graph produced by the modular program.

To automatically search these 120 models for high-quality options, we must first define ``quality'' by choosing a model-scoring metric.
One reasonable approach is to measure a model's predictive accuracy by computing its Expected Log-Posterior Density (ELPD).
ELDP approximates the leave-one-out prediction accuracy of a model for a dataset \cite{elpd}.

While ELPD is relatively efficient, it could still take a long time to accurately compute ELPD for every model in our network.
Our goal, then, is to find high-quality models with as few ELPD evaluations as possible.
Here we take the simplest approach, a greedy graph search, and leave more sophisticated search methods for future work.

Our greedy graph search loops over the following steps, given an arbitrary starting point:
\begin{enumerate}
\item Score the neighbors of the current model.
\item Move to the highest scoring model seen so far, or if that is the current model, return it.
\end{enumerate}
This search algorithm is not guaranteed to be optimal; it is analogous to a gradient descent of a (likely non-convex) space.
Intuitively, the closer the maximal-neighbor operation is to a gradient, the more efficient the search will be, so we can expect that the semantically minimal changes represented by the edges in the model graph provide exactly the topology of models closest to a smooth continuous space.

The search algorithm is implemented as a short Python script that uses the prototype compiler's implementation of the \texttt{ModelNeighbors} algorithm described in \cref{neighbors-algorithm}.
Its source code is available online \cite{mstan-github}.

Starting from the case study's first model, the greedy graph search followed the path shown in \cref{fig:searchpath}:

\begin{figure}[htbp]
\centering
\includegraphics[width=4.5in]{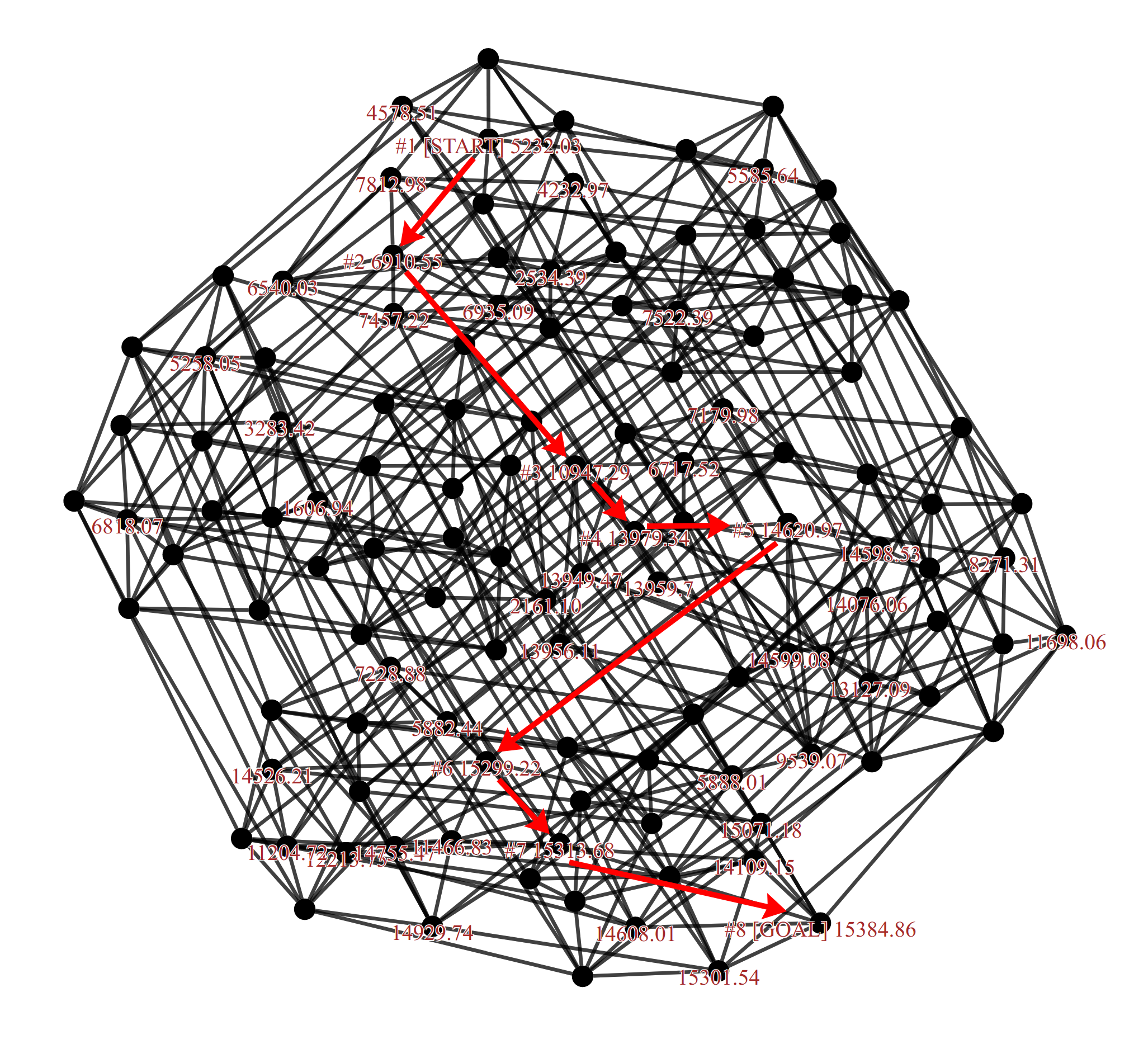}
\caption{\label{fig:searchpath}The red annotations show the ELPD scores for the assessed models. The search algorithm visited nodes along the red arrow path from starting at \texttt{[START]} and terminating at \texttt{[GOAL]}.}
\end{figure}

The search performed 47 ELPD evaluations.
The search agreed with the case study authors' final model, confirming that it has (at least locally) optimal predictive performance.

While greedy maximization of ELPD is a naive statistical workflow and shouldn't be blindly trusted it to give a final model, it is at least useful for finding promising neighborhoods, especially for large model spaces.
\section{Future work}
\label{sec:org01deddb}
\label{future-applications}
As discussed in \cref{uses}, there are many motivating use cases for the network of models.
We hope to demonstrate and build tooling for more of these use cases, using modular Stan as a foundation:
\begin{itemize}
\item Model search.
\Cref{birthday} gave a simple algorithm for greedy search that only used the network topology.
 We hypothesize that more efficient search methods could leverage the module structure, treat search as an exploration/exploitation problem, or utilize methods from the symbolic regression literature.
\item Model space navigation tools.
Model developers must decipher promising directions in which to iterate.
If we could annotate the edges between models with meaningful joint metrics, we could provide developers with direction.
We could also use model statistics or edge weights to embed network-space into visualizations.
\item Multi-model ensemble methods.
As mentioned in \cref{uses}, multi-model ensemble methods like stacking and multiverse analysis could be applied to multi-model programs.
\item Sensitivity analysis.
To automate sensitivity analysis given a model and some metric of a model's results, like a \(p\text{-value}\), we can check the extent to which the model's results differ from its neighbors.
\end{itemize}
\label{additional-features}
In addition, we hope to expand the capabilities of the swappable module system.
\begin{itemize}
\item Explicit model signatures.
Thought we believe that our module signature inference scheme is more beginner-friendly, explicit signatures would make module reuse across applications easier.
\item Implementation co-selection logic.
While features like module fields allow users to encode the constraint that a set of implementations should always be selected together, there is no reason that user constraints on selection sets should not be arbitrarily complex, the end point being a predicate logic.
\end{itemize}
\pagebreak
\bibliographystyle{unsrtnat}
\begin{footnotesize}

\end{footnotesize}
\end{document}